\documentclass[11pt]{article} 

\usepackage[utf8]{inputenc} 
\usepackage{authblk}

\usepackage{geometry} 
\usepackage{graphicx} 

\usepackage{booktabs} 
\usepackage{array} 
\usepackage{verbatim} 
\usepackage{subfig} 
\usepackage{caption}
\usepackage{slashbox,amssymb,amsthm}
\usepackage{amsmath,amsfonts,graphicx,enumitem}
\providecommand{\norm}[1]{\lVert#1\rVert}
\usepackage{color, float}
\usepackage{xr}

\usepackage{natbib}
\DeclareMathOperator*{\argmin}{arg\,min}
\usepackage[colorlinks=true, a4paper=true, pdfstartview=FitV,
linkcolor=blue, citecolor=blue, urlcolor=blue]{hyperref}
\newcounter{myalgctr}

\newenvironment{rem}{
   \vskip1mm\indent
   \refstepcounter{myalgctr}
   \textbf{Remark \themyalgctr}
   }{\hfill$\diamond$\par}  
\numberwithin{myalgctr}{section}
\numberwithin{equation}{section}
\newtheorem{defi}{Definition}[section]
\newtheorem{thm}{Theorem}[section]
\newtheorem{conj}{Conjecture}[section]

\newcommand{\bl}{{\bar{\lambda}}}
\usepackage{fancyhdr} 
\usepackage{sectsty}
\allsectionsfont{\sffamily\mdseries\upshape} 

\usepackage[nottoc,notlof,notlot]{tocbibind} 
\usepackage[titles,subfigure]{tocloft} 

\begin{document}
\title{Statistical Inference based on Bridge Divergences}
\author[1]{Arun Kumar Kuchibhotla\thanks{arunku@upenn.edu}}
\author[1]{Somabha Mukherjee\thanks{somabha@upenn.edu}}
\author[2]{Ayanendranath Basu\thanks{ayanbasu@isical.ac.in}}
\affil[1]{Department of Statistics, University of Pennsylvania}
\affil[2]{Interdisciplinary Statistical Research Unit, Indian Statistical Institute}


\date{}
\maketitle
\begin{abstract}
$M$-estimators offer simple robust alternatives to the maximum likelihood estimator. Much of the robustness literature, however, has focused on the problems of location, location-scale and regression estimation rather than on estimation of general parameters. The density power divergence (DPD) and the logarithmic density power divergence (LDPD) measures provide two classes of competitive $M$-estimators (obtained from divergences) in general parametric models which contain the MLE as a special case. In each of these families, the robustness of the estimator is achieved through a density power down-weighting of outlying observations. Both the families have proved to be very useful tools in the area of robust inference. However, the relation and hierarchy between the minimum distance estimators of the two families are yet to be comprehensively studied or fully established. Given a particular set of real data, how does one choose an optimal member from the union of these two classes of divergences? In this paper, we present a generalized family of divergences incorporating the above two classes; this family provides a smooth bridge between the DPD and the LDPD measures. This family helps to clarify and settle several longstanding issues in the relation between the important families of DPD and LDPD, apart from being an important tool in different areas of statistical inference in its own right.
\end{abstract}
\section{Introduction}\label{sec:intro}
Statistical procedures based on minimization of divergences are popular in the literature. In our context, a divergence is a distance like dissimilarity measure between two distributions which does not necessarily demand metric properties. Density based minimum divergence methods are very useful as they combine high efficiency with strong robustness properties. \cite{bhhj98} proposed the class of density power divergences (DPD). These divergences, when constructed between an arbitrary density $g$ and a parametric model density $f_{\theta}$ are indexed by a non-negative robustness tuning parameter $\alpha$ and have the form
\begin{equation}\label{type1}
\rho^{(\alpha)}_1(g,f_{\theta})=\int \left\{f_{\theta}^{1+\alpha}-\left(1+\frac{1}{\alpha}\right)gf_{\theta}^{\alpha} + \frac{1}{\alpha}g^{1+\alpha}\right\}.
\end{equation}
The divergence at $\alpha = 0$ is defined as the continuous limit of the expression in Equation \eqref{type1} as $\alpha\to0$. This generates the divergence
\begin{equation}\label{eq:KLDiver}
\rho^{(0)}_1(g,f_{\theta})=\int g\log\left(\frac{g}{f_{\theta}}\right),
\end{equation}
where $\log$ represents natural logarithm. Observe that the only term on the right hand side of Equation \eqref{type1} containing both $g$ and $f_{\theta}$ has $g$ in degree one. A divergence with such a property has been referred to as a decomposable divergence by some authors; see, for example, \cite{bt12}. 

Another class of divergences with some similar properties was introduced by \cite{jhhb01} and was followed up by \cite{fe08}, \cite{bt12} and \cite{fuji13} among others. In spite of its formal similarity with the DPD, this family, referred to herein as the logarithmic density power divergence (LDPD) family, was originally developed following the robust model fitting idea of \cite{wind95}. This class also has a non-negative robustness tuning parameter $\alpha$ and is defined as
\begin{equation}\label{type0}
\rho^{(\alpha)}_0(g,f_{\theta})=\log\int f_{\theta}^{1+\alpha}-\left(1+\frac{1}{\alpha}\right)\log\int gf_{\theta}^{\alpha}+\frac{1}{\alpha}\log\int g^{1+\alpha}.
\end{equation}
The divergence $\rho_0^{(0)}$ is once again defined as a limiting case and some simple algebra shows that $\rho_0^{(0)}(g,f_{\theta}) = \rho_1^{(0)}(g,f_{\theta})$, the common divergence being a version of the Kullback-Leibler divergence.  

\cite{jhhb01} also provided a general form of divergence measures in terms of two tuning parameters $\alpha$ and $\phi$ given by
\begin{equation}\label{typephi}
\rho^{(\alpha)}_{\phi}(g,f_{\theta}) = \frac{1}{\phi}\left(\int f_{\theta}^{1+\alpha}\right)^{\phi}-\frac{1}{\phi}\left(1+\frac{1}{\alpha}\right)\left(\int gf_{\theta}^{\alpha}\right)^{\phi} + \frac{1}{\alpha\phi}\left(\int g^{1+\alpha}\right)^{\phi},
\end{equation}
for $0\le\phi\le1$ and $\alpha \ge 0$. The DPD and the LDPD measures are recovered from this general form for $\phi = 1$ and $\phi = 0$ , the second one being obtained as the limiting case as $\phi \rightarrow 0$. Accordingly, the DPD and the LDPD families have also been referred to as type 1 and type 0 divergences in the literature. A scaled version of the LDPD has also been called the $\gamma$-divergence by \cite{fe08}. Both the DPD and the LDPD measures have proved to be useful additions to the literature of robust parameter estimation based on divergences. Both families, as well as the corresponding minimum divergence estimators, have been heavily cited in the literature and applied to many practical problems, and a thorough exploration of their relationship and hierarchy can help us develop better compromises. In either case, a simple density power down-weighting indicates the source of robustness of the resulting estimator. In both cases, the parameter $\alpha$ controls the trade-off between robustness and efficiency; smaller values of $\alpha$ lead to greater model efficiency and larger values of $\alpha$ lead to greater outlier stability. The minimum divergence estimator corresponding to $\alpha = 0$ is the maximum likelihood estimator in either case.

In terms of comparison between the minimum DPD and the minimum LDPD estimators, \cite{jhhb01} expressed a (weak) preference for the former. In particular, they observed that the presence of observations very close to zero could lead to spurious global minimum in case of the LDPD under the exponential model. On the other hand, \cite{fe08} and \cite{fuji13} have claimed that the minimum LDPD estimators exhibit a greater relative stability under heavy contamination leading to smaller bias and smaller mean squared error. They argued that the small bias that the minimum LDPD estimator has is due to an approximate Pythagorean relation that the LDPD satisfies. In contrast, \cite{bt12}, based on their own simulation study, do not report any particular relative advantage for the minimum LDPD estimator. 

These points raise certain unsettled issues involving these two useful classes of divergences and corresponding estimators which we hope to at least partially reconcile in the present paper. The generalized family of divergences in Equation \eqref{typephi}, useful as it is, does not generate minimum divergence estimators that are legitimate $M$-estimators except when $\phi=0$ or $1$. We will construct an alternative generalized class of divergences providing a bridge between these classes where several of the intermediate divergences lead to reasonable compromises between the positives of these two families. This new family will be called the family of \emph{bridge density power divergences} and the generated estimators are \underline{all} $M$-estimators. This will provide the user with the flexibility of choosing a suitable estimator from a larger class.

The new family of divergences will depend on two tuning parameters which are (i) the robustness parameter $\alpha$ and (ii) the bridge parameter $\lambda$. It would be of interest to choose the tuning parameters adaptively with respect to the proportion of outliers so that the estimation is optimal in an appropriate sense. The robustness tuning parameter $\alpha$ should be close to zero for pure data and should assume a moderately large positive value in case of contaminated data. In this respect, \cite{HK01} proposed the first method of choosing the tuning parameter by minimizing an estimate of the asymptotic variance and \cite{wj05} refined this process by using the mean squared error criterion together with a pilot estimate. We provide some justification for using a modified \cite{HK01} procedure; see Section \ref{tunpar}.

Before concluding this section, we summarize the main achievements of this paper.
\begin{enumerate}
\item We introduce a new family of divergences, the \emph{Bridge Density Power Divergences} (BDPD) which produce a smooth link between the DPD and the LDPD; each intermediate divergence is decomposable and leads to an $M$-estimator. Apart from helping to understand the relations between the DPD and the LDPD, this family is important in its own right, and in specific cases the performance of the intermediate divergences turn out to be superior than both the two marginal divergences (DPD and LDPD).
\item We demonstrate that the spurious root problem observed in case of the LDPD as noticed by \cite{jhhb01} is not an isolated problem for the exponential model, and is a more general phenomenon.
\item It is shown that the results and assertions of \cite{fe08} and \cite{fuji13} are substantially correct, but only when a number of residual concerns are answered. In particular, the observed superior performance of the ``minimum LDPD estimator", is, most of the time, achieved at a local minimum of the LDPD measure (rather than a global one), and has to be viewed as a weighted method of moments estimator rather than a minimum divergence estimator.
\item We demonstrate that the weighted method of moment equations corresponding to the LDPD (and indeed many other BDPD measures) can potentially throw up multiple roots; this is a real problem with practical implications. As the desired root is not necessarily the minimizer of the corresponding divergence, there is no automatic root selection strategy. As a choice of the incorrect solution can be disastrous, it is imperative that a clear answer to this question is provided in the literature, where this issue is so far unexplored.
\item We clearly define the target parameter when using the LDPD for parametric estimation. Our results in this paper show that it cannot, in general, be the global minimizer of the LDPD, nor any arbitrary root of the weighted method of methods equation. Our description of the target depends on an actual numerical algorithm for its selection (see item 6 below). This interpretation of the target parameter actually extends to the entire BDPD family, except the ordinary DPD. 
\item Finally, we provide an algorithm for the selection of the suitable root of the weighted method of moment equations for any member of the BDPD family, including the LDPD. We also provide a method for the selection of ``optimal'' tuning parameters for analyzing the real data. 
\end{enumerate}

The rest of the paper is organized as follows. In Section \ref{div}, we introduce the basic parametric set up and also provide the construction leading to the family of bridge divergences. Various properties of the estimators including strong consistency, asymptotic normality and robustness are discussed in Section \ref{prop}. In Section \ref{type0s}, we provide a heuristic explanation for the spurious minimum behaviour of the LDPD measures in the case of small outliers. Section \ref{sec:Unbounded} rigorously proves that the LDPD is bound to fail in some specific examples. In Section \ref{sec:Chain_Algorithm}, we address issues regarding multiple roots of the LDPD estimating equation and well-definedness of the estimators as discussed in the existing literature. We also introduce an algorithm for getting hold of a good root of the LDPD and all the other BDPD estimating equations. In Section \ref{sec:Simulations}, we provide the results of a simulation study conducted to support the claims made in previous sections. In Section \ref{tunpar}, we give some directions on how to choose the tuning parameters. Finally we conclude with a few remarks in Section \ref{con}.
\section{The Bridge Density Power Divergence (BDPD) Family}\label{div}
The problem of parameter estimation considered in this paper uses the following set up and notation. Let $\mathcal{G}$ denote the set of all probability distributions having densities with respect to some $\sigma$-finite base measure $\mu$. We assume that the data generating distribution $G$ and the model family $\mathcal{F}_{\Theta}=\{F_{\theta}:\theta\in\Theta\subset\mathbb{R}^p\}$ belong to $\mathcal{G}$. Let $g$ and $f_{\theta}$ be the corresponding densities (with respect to $\mu$). Let $X_1,X_2,\ldots,X_n$ be an independent and identically distributed (iid) random sample from $G$ which is to be modeled by the family $\mathcal{F}_{\Theta}$. Our aim is to estimate the parameter $\theta$ by choosing the model density which gives the ``closest fit" to the data. In this paper, we quantify the ``closeness" using the density power divergences, the logarithmic density power divergences or their generalizations as mentioned in the previous section. However we will see in the latter sections that this idea of closeness will require a refinement for most of the minimum BDPD estimators. All the integrals considered in this paper including those in the previous section are with respect to the measure $\mu.$

In order to estimate $\theta$ based on an iid sample, one needs to construct an empirical estimate of the divergence. In this regard, note that the first terms of Equations \eqref{type1} and \eqref{type0}, depending only on the model density $f_{\theta}$, need no estimation. The third terms are independent of $\theta$ and so do not figure in any minimization process over $\theta$. For the second terms, note that the integral involved can be written as $\mathbb{E}_gf_{\theta}^{\alpha}(X)$ and so can be estimated by the average of $f_{\theta}^{\alpha}(X_i),\, 1\le i\le n$.
This is a consequence of the decomposability property referred to in Section \ref{sec:intro}. In light of this discussion, the parameter estimates based on the divergences $\rho^{(\alpha)}_1$ and $\rho^{(\alpha)}_0$ are given by
\begin{equation}\label{eq:MDPDE}
\hat{\theta}^{\alpha}_{n1} := \argmin_{\theta\in\Theta}\left[ \int f_{\theta}^{1+\alpha} - \left(1+\frac{1}{\alpha}\right)\frac{1}{n}\sum_{i=1}^nf_{\theta}^{\alpha}(X_i)\right],
\end{equation}
and
\begin{equation}\label{eq:LDPDSampleObj}
\hat{\theta}^{\alpha}_{n0} := \argmin_{\theta\in\Theta} \left[\log\left(\int f_{\theta}^{1+\alpha}\right) - \left(1+\frac{1}{\alpha}\right)\log\left(\frac{1}{n}\sum_{i=1}^nf_{\theta}^{\alpha}(X_i)\right)\right].
\end{equation}
Under certain regularity conditions allowing the interchange of the derivative and the integral, the estimating equations corresponding to the above divergences are given, respectively, by
\begin{equation}\label{eq1}
\int f_{\theta}^{1+\alpha}u_{\theta} = \frac{1}{n}\sum_{i=1}^n f^{\alpha}_{\theta}(X_i)u_{\theta}(X_i),
\end{equation}
and
\begin{equation}\label{eq2} 
\left(\frac{1}{n}\sum_{i=1}^n f_{\theta}^{\alpha}(X_i)\right)\int f_{\theta}^{1+\alpha}u_{\theta} = \left(\frac{1}{n}\sum_{i=1}^n f_{\theta}^{\alpha}(X_i)u_{\theta}(X_i)\right)\int f_{\theta}^{1+\alpha}. 
\end{equation}
Here $u_{\theta}(x)$ represents the likelihood score function given by $u_{\theta}(x) = \nabla \log f_{\theta}(x)$. Throughout this manuscript, we use the symbols $\nabla, \nabla_2$ to denote the first and the second derivatives with respect to $\theta$. Equations \eqref{eq1} and \eqref{eq2} demonstrate that both the minimum DPD and the minimum LDPD estimators are legitimate $M$-estimators. Later in this section and in Section \ref{sec:Chain_Algorithm}, we will consider the weighted method of moments representation of these estimating equations.

In this paper, we will show that the DPD and the LDPD families can be combined in a larger super-family of divergences where each intermediate divergence class leads to an $M$-estimator unlike the generalization given in Equation \eqref{typephi}. For notational simplicity in the following derivation, define 
\[
t_1(\theta):=\int f_{\theta}^{1+\alpha},\quad t_2(\theta):=\int gf_{\theta}^{\alpha},\quad t_3(\theta):=\int g^{1+\alpha}.
\]

We consider the density versions (in terms of the true density $g$) of Equations \eqref{eq1} and \eqref{eq2} rather than the versions based on the empiricals. These equations are given by
\begin{equation}\label{seq1}
\int f_{\theta}^{1+\alpha}u_{\theta} = \int f_{\theta}^{\alpha}gu_{\theta},
\end{equation}
and
\begin{equation}\label{seq2}
\int f_{\theta}^{\alpha}g\int f_{\theta}^{\alpha}u_{\theta} = \int f_{\theta}^{\alpha}gu_{\theta}\int f_{\theta}^{\alpha}.
\end{equation}
Rewriting the estimating equations \eqref{seq1} and \eqref{seq2} in terms of $t_1$ and $t_2$, we get
\begin{align}
\left[\frac{\nabla t_1(\theta)}{\alpha+1}-\frac{\nabla t_2(\theta)}{\alpha}\right]&=0,\label{eq3}\\
\left[\frac{t_2(\theta)\nabla t_1(\theta)}{\alpha+1}-\frac{t_1(\theta)\nabla t_2(\theta)}{\alpha}\right]&=0\label{eq4}.
\end{align}
For $\lambda\in[0,1]$, consider an estimating equation which equates a convex combination of the above two estimating functions to zero; it is given explicitly by
\begin{equation}\label{besteq}
\lambda\left[\frac{t'_1(\theta)}{\alpha+1}-\frac{t'_2(\theta)}{\alpha}\right]+(1-\lambda)\left[\frac{t_2(\theta)t'_1(\theta)}{\alpha+1}-\frac{t_1(\theta)t'_2(\theta)}{\alpha}\right]=0.
\end{equation}
Rearranging the terms on the left hand side leads to the differential equation,
\begin{align}
\frac{t'_1(\theta)}{1+\alpha}[\lambda+\bar{\lambda} t_2(\theta)]&=\frac{t'_2(\theta)}{\alpha}[\lambda+\bar{\lambda} t_1(\theta)],\nonumber\\
\frac{t'_1(\theta)}{\lambda+\bar{\lambda} t_1(\theta)} &=\left(\frac{1+\alpha}{\alpha}\right)\frac{t'_2(\theta)}{\lambda+\bar{\lambda} t_2(\theta)},\label{momenteqold}
\end{align}
where $\bar{\lambda}=1-\lambda$. It is now easy to derive the corresponding objective function which turns out to be 
\begin{equation}
\frac{1}{\bar{\lambda}}\log(\lambda+\bar{\lambda} t_1(\theta))-\frac{1}{\bar{\lambda}}\left(\frac{1+\alpha}{\alpha}\right)\log(\lambda+\bar{\lambda} t_2(\theta))+C,
\end{equation}
for some constant $C$ independent of $\theta$.
Imposing the condition that the divergence has to be zero when $g = f_{\theta}$, the constant $C$ is recovered to be $\frac{1}{\bar{\lambda}}\frac{1}{\alpha}\log(\lambda+\bar{\lambda}t_3(\theta))$. Hence the new divergence can be written as 
\begin{equation}\label{eq:BridgeDivergence}
\rho^{(\alpha,\lambda)}(g,f_{\theta})=\frac{1}{\bar{\lambda}}\log(\lambda+\bar{\lambda} t_1(\theta))-\frac{1}{\bar{\lambda}}\left(\frac{1+\alpha}{\alpha}\right)\log(\lambda+\bar{\lambda} t_2(\theta))+\frac{1}{\bar{\lambda}}\frac{1}{\alpha}\log(\lambda+\bar{\lambda} t_3(\theta)).
\end{equation}
This is our class of \emph{bridge density power divergences}, defined for $\alpha\ge0$ and $\lambda\in[0,1]$. For $\lambda = 0$, this reduces to the class of logarithmic density power divergences; as $\lambda\to1$, we recover the class of density power divergences. It is also easy to verify that as $\alpha\rightarrow0$, the limiting divergence is the Kullback-Leibler divergence as given in Equation \eqref{eq:KLDiver} irrespective of the value of $\lambda$. Due to the consideration of efficiency, we will, in the rest of the paper, restrict the robustness parameter $\alpha$ to the range $[0,1]$.

The estimating equation \eqref{momenteqold} for the bridge divergence with parameters $(\alpha, \lambda)$ can be written as:
\begin{equation}\label{momentequation}
	\frac{\int f_{\theta}^{1 + \alpha}u_{\theta}}{\lambda + \bar{\lambda}\int f_{\theta}^{1+\alpha}} = \frac{\frac{1}{n}\sum_{i=1}^n f_{\theta}^{\alpha}(X_i)u_{\theta}(X_i)}{\lambda + \bar{\lambda}\frac{1}{n}\sum_{i=1}^n f_{\theta}^{\alpha}(X_i)}.
\end{equation}
This may be referred to as the BDPD moment equations. It is immediately obvious that the above equation represents a very rich class of score moment equations. When $\alpha = 0$, the above represents the ordinary likelihood score equation irrespective of the value of $\lambda$. More generally it represents a score moment equation where the scores are variably weighted by powers of densities depending on the parameter $\alpha$, and the weights are variably normalized depending on the parameter $\lambda$. When $\lambda = 0$, the weights are perfectly normalized (as in the case of the LDPD), while we get the non-normalized equations for $\lambda = 1$, as in the case of the DPD. For intermediate values of $\lambda$ we get partially normalized estimating equations, with the degree of normalization dropping off with increasing $\lambda$. We will observe the effect of this normalization throughout the rest of the paper, most notably in Section \ref{type0s}, where we will give some indication of its role in the occurrence of the spurious roots. 
\begin{rem}
The BDPD family was constructed using an appropriate combination of the DPD and the LDPD estimating equations with the same tuning parameter $\alpha$; there is no specific reason, apart from mathematical convenience, for considering the same $\alpha$ in both the estimating equations. However, if estimating equations with different values of $\alpha$ are combined it does not appear to lead to a genuine objective function for such an estimating equation. 
\end{rem}
\begin{rem}
Equation \eqref{besteq} cannot be obtained by starting directly with a convex combination of the DPD and the LDPD with the same value of $\alpha$.
\end{rem}
\section{Properties of Minimum BDPD Estimators}\label{prop}
Based on the bridge density power divergence defined in the previous section and an iid random sample $X_1,X_2,\ldots,X_n$ from a distribution $G$, an estimator of $\theta$ is given by
\begin{align}
\hat{\theta}^{(\alpha,\lambda)}_n := \argmin_{\theta\in\Theta}\frac{1}{\bar{\lambda}}\log(\lambda+\bar{\lambda} t_1(\theta))-\frac{1}{\bar{\lambda}}\left(\frac{1+\alpha}{\alpha}\right)\log\left(\lambda+\bar{\lambda} \frac{1}{n}\sum_{i=1}^n f_{\theta}^{\alpha}(X_i)\right).\label{bridge}
\end{align}
By definition the BDPD leads to an $M$-estimator and so its asymptotic and robustness properties can be derived from the well-established $M$-estimation theory. Therefore, we present the results for consistency and asymptotic normality of the bridge divergence estimator without proofs. For this, we use the following notation.
\[
P_{i,\gamma} = \int g^{i}f_{\theta}^{\gamma},\;\;Q_{i,\gamma} = \int g^if_{\theta}^{\gamma}u_{\theta}u_{\theta}^{\top},\;\;R_{i,\gamma} = \int g^if_{\theta}^{\gamma}u_{\theta},\;\;S_{i,\gamma} = \int g^if_{\theta}^{\gamma}\nabla u_{\theta}.
\]

We now present a consistency theorem with conditions in alignment with those in Wald's consistency theorem (\citet[Chapter 17]{Ferguson}) for the maximum likelihood estimator and as the proof is essentially the same, it is moved to the Supplementary Material.
\setcounter{thm}{0}
\begin{thm}[Consistency]\label{Consistent}
Suppose that the following assumptions hold.
\begin{enumerate}[label = \bfseries(C\arabic*)]
\item $\Theta$ is a compact metric space;\label{c1}
\item There exists a function $K(x)$ (independent of $\theta$) such that $|\varphi_{\alpha}(x)|\le K(x)$ and $K(X)$ has finite expectation (with respect to $G$). Here $\varphi_{\alpha}(x) = f_{\theta}^{\alpha}(x)$ for $\alpha > 0$ and $\varphi_0(x) = \log f_{\theta}(x)$;\label{c2}
\item For each $x$ and any sequence $\theta_n\to\theta$, $$\lim_{n\to\infty}f_{\theta_n}(x) = f_{\theta}(x),$$ for all $x$, except possibly on a set (which might depend on $\theta$ but not on the sequence $\{\theta_n\}$) of $\mu$-measure zero. Also, $\theta_g^{(\alpha,\lambda)}$ is the unique minimizer of the bridge density power divergence $\rho^{(\alpha,\lambda)}(g,f_{\theta})$. \label{c3}
\end{enumerate}
Then the minimum BDPD estimator $\hat{\theta}_n^{(\alpha,\lambda)}$ is strongly consistent for $\theta_g^{(\alpha,\lambda)}$ for any fixed $\alpha \in [0,1]$ and $\lambda > 0$.
\end{thm}
\begin{rem}
Note that the conditions are independent of the value of $\lambda$ and match those of Wald's consistency theorem when $\alpha = 0$. The consistency theorem only requires $\Theta$ to be a metric space and not an Euclidean space. The proof can be extended to settings other than iid samples using various generalizations of the uniform strong law of large numbers (USLLN).
\end{rem}
\begin{rem}
The main ingredient in the proof of Theorem \ref{Consistent} is Theorem 5.7 of \cite{VAAR98} that uses uniform convergence of the sample based objective functions to their population counterparts. For $\lambda > 0$, the function $\xi\mapsto \log(\lambda + \bar{\lambda}\xi)$ is Lipschitz on any bounded closed interval $I$ that does not contain $0$. This fact allows one to conclude the following implication (under \ref{c2}):
\begin{align*}
&\sup_{\theta\in\Theta}\left|\frac{1}{n}\sum_{i=1}^n f_{\theta}^{\alpha}(X_i) - \int gf_{\theta}^{\alpha}\right| = o_p(1),\quad(\mbox{follows from uniform SLLN})\\
\Rightarrow&\sup_{\theta\in\Theta}\left|\log\left(\lambda + \bar{\lambda}\frac{1}{n}\sum_{i=1}^n f_{\theta}^{\alpha}(X_i)\right) - \log\left(\lambda + \bar{\lambda}\int gf_{\theta}^{\alpha}\right)\right| = o_p(1).
\end{align*}
At $\lambda = 0$ the Lipschitz property fails and the above implication breaks down. This is the reason for restricting the range of $\lambda$ to $\lambda > 0$. Thus the above proof does not automatically cover the DPD family. To extend the consistency claim to $\lambda = 0$, one needs to make the additional assumption $\inf_{\theta\in\Theta} \int gf_{\theta}^{\alpha} > 0.$
\end{rem}
The proof of asymptotic normality of the minimum BDPD estimator can be obtained through a quadratic approximation of the objective function. Using Cram$\acute{e}$r-Rao type regularity conditions, we also observe that there exists a sequence of roots of the estimating equation which is consistent and asymptotically normal. We do not repeat the conditions here but refer the reader to Theorem 2 of \cite{bhhj98}. The proof is similar to Lehmann's proof (\cite{LEH}) of consistency and asymptotic normality of the MLE, and is hence omitted.  
\begin{thm}[Asymptotic Normality]\label{asympnorm}
Under certain regularity conditions, there exists a sequence of roots $\theta_n$ of the bridge divergence estimating equation which is consistent and $\sqrt{n}({\theta}_n - \theta_g^{(\alpha,\lambda)})$ has an asymptotically normal distribution with mean 0 and variance given by $J(\theta_g)^{-1}K(\theta_g)J(\theta_g)^{-1}$, where $\theta_g := \theta_g^{(\alpha, \lambda)}$,
\begin{align*}
K(\theta) &= (1-\lambda)^2P_{1,2\alpha}R_{0,\alpha+1}R_{0,\alpha+1}^{\top} - \bar{\lambda}\left[\lambda + \bar{\lambda}P_{0,\alpha+1}\right]R_{0,\alpha+1}R_{1,2\alpha}^{\top}\\ &\quad-\bar{\lambda}\left[\lambda + \bar{\lambda}P_{0,\alpha+1}\right]R_{1,2\alpha}R_{0,\alpha+1}^{\top} + \left[\lambda + \bar{\lambda}P_{0,\alpha+1}\right]^2Q_{1,2\alpha}\\
&\quad-\bar{\lambda}^2P_{1,\alpha}^2R_{0,\alpha+1}R_{0,\alpha+1}^{\top} + \bar{\lambda}\left[\lambda + \bar{\lambda}P_{0,\alpha+1}\right]P_{1,\alpha}R_{1,\alpha}R_{0,\alpha+1}^{\top}\\&\quad+\bar{\lambda}\left[\lambda + \bar{\lambda}P_{0,\alpha+1}\right]P_{1,\alpha}R_{0,\alpha+1}R_{1,\alpha}^{\top} - \left[\lambda + \bar{\lambda}P_{0,\alpha+1}\right]^2R_{1,\alpha}R_{1,\alpha}^{\top}\\
J(\theta) &= (\alpha + 1)Q_{0,\alpha+1}\left[\lambda + \bar{\lambda}P_{1,\alpha}\right] + S_{0,\alpha + 1}\left[\lambda + \bar{\lambda}P_{1,\alpha}\right]\\
&\quad-\alpha Q_{1,\alpha}\left[\lambda + \bar{\lambda}P_{0,\alpha+1}\right] - S_{1,\alpha}\left[\lambda + \bar{\lambda}P_{0,\alpha + 1}\right]\\
&\quad +\bar{\lambda}\alpha R_{1,\alpha}R_{0,\alpha+1}^{\top} - \bar{\lambda}(\alpha+1)R_{0,\alpha+1}R_{1,\alpha}^{\top}.
\end{align*}
\end{thm}
\subsection{Pythagorean Relation}
\cite{fe08} have claimed that under heavy contamination, the minimum LDPD estimator achieves a smaller bias (compared to the minimum DPD estimator); it is also suggested that this phenomenon can be partially explained by an approximate Pythagorean relation which the LDPD satisfies. We present a sequence of results which shows that such a Pythagorean relation holds in general for all BDPD (except the DPD), so that the LDPD result is a special case of the more general result. For any $0\le t\le 1$, let $\bar{t} = 1-t$. Define the cross entropy between any two densities $g$ and $f$ for $\lambda,\alpha\in[0,1]$ by,
\[
d_{\lambda,\alpha}(g,f) = \frac{1}{\bar{\lambda}}\frac{1}{1+\alpha}\log\left(\lambda + \bar{\lambda}\int f^{1+\alpha}\right) - \frac{1}{\alpha\bar{\lambda}}\log\left(\lambda + \bar{\lambda}\int gf^{\alpha}\right).
\]
The divergence induced by this cross entropy is given by
\[
D_{\lambda,\alpha}(g,f) = -d_{\lambda,\alpha}(g,g) + d_{\lambda,\alpha}(g,f).
\]
which is a scaled version of $\rho^{\lambda,\alpha}(g,f)$, satisfying $D_{\lambda,\alpha}(\cdot,\cdot) = \rho^{\lambda,\alpha}(\cdot,\cdot)/(1+\alpha)$.
\begin{thm}
Let $f$ and $\delta$ be given probability density functions and let $0\le\varepsilon\le 1$. If $g(\cdot) = (1-\varepsilon)f(\cdot) + \varepsilon\delta(\cdot)$ and $h$ is any positive function, then for any $\alpha\in[0,1]$ and $\lambda\in[0,1)$,
\[
d_{\lambda,\alpha}(g,h) = d_{\lambda,\alpha}(f,h) - \frac{1}{\alpha\bar{\lambda}}\log(1-\varepsilon) + O(T_{\epsilon,\delta}),
\]
where
\[
T_{\varepsilon,\delta} := \frac{\varepsilon}{1-\varepsilon}\left[\lambda + \bar{\lambda}\int \delta h^{\alpha}\right]\bigg/\alpha\bar{\lambda}\left[\lambda+\bar{\lambda}\int fh^{\alpha}\right].
\]
For $\lambda = 1$, we have
\[
d_{1,\alpha}(g,h) = d_{1,\alpha}(f,h) + \frac{\varepsilon}{\alpha}\left[\int \{f-\delta\}h^{\alpha}\right]. 
\]
\end{thm}
\begin{proof}
First consider the case $0\le \lambda < 1$. By definition of $d_{\lambda,\alpha}$, we have, using $\ell(h) = \frac{1}{\bar{\lambda}(1+\alpha)}\log\left(\lambda + \bar{\lambda}\int h^{1+\alpha}\right)$,
\begin{align*}
d_{\lambda,\alpha}(g,h) &= \ell(h) - \frac{1}{\bar{\lambda}\alpha}\log\left(\lambda + \bar{\lambda}\int \{(1-\varepsilon)f + \varepsilon\delta\}h^{\alpha}\right)\\
&= \ell(h) - \frac{1}{\bar{\lambda}\alpha}\log\left(\lambda + \bar{\lambda}\int(1-\varepsilon)fh^{\alpha} + \bar{\lambda}\int \varepsilon\delta h^{\alpha}\right)\\
&= \ell(h) - \frac{1}{\bar{\lambda}\alpha}\log\left(\varepsilon\left[\lambda + \bar{\lambda}\int \delta h^{\alpha}\right] + (1-\varepsilon)\left[\lambda + \bar{\lambda}\int fh^{\alpha}\right]\right).
\end{align*}
Now by a simple Taylor series expansion, we get that
\[
d_{\lambda,\alpha}(g,h) = \ell(h) - \frac{1}{\alpha\bar{\lambda}}\log\left((1-\varepsilon)\left[\lambda + \bar{\lambda}\int fh^{\alpha}\right]\right) + O(T_{\varepsilon,\delta}),
\] 
from which the result follows. For the case $\lambda = 1$, observe that
\begin{align*}
d_{1,\alpha}(g,h) &= \frac{1}{1+\alpha}\int h^{1+\alpha} - \frac{1}{\alpha}\int gh^{\alpha}\\
&= \frac{1}{1+\alpha}\int h^{1+\alpha} - \frac{1}{\alpha}\int fh^{\alpha} + \frac{\varepsilon}{\alpha}\left[\int fh^{\alpha} - \int \delta h^{\alpha}\right]\\
&= d_{1,\alpha}(f,h) + \frac{\varepsilon}{\alpha}\left[\int \{f-\delta\} h^{\alpha}\right].
\end{align*}
\end{proof}
\begin{rem}
This theorem can be used to prove a Pythagorean relation (for $0\le \lambda < 1$) similar to Theorem 3.2 of \cite{fe08}. The statement of the result is given below for completeness; however the proof is very similar to the one in \cite{fe08} and is hence omitted. The \cite{fe08} result thus becomes a particular case of the following theorem.
\end{rem}
\begin{thm}[Pythagorean Relation]
Let $f$ and $\delta$ be given probability density functions and let $0\le\varepsilon\le 1$. If $g(\cdot) = (1-\varepsilon)f(\cdot) + \varepsilon\delta(\cdot)$ and $h$ is any positive function, then for any $0\le \lambda < 1$,
\[
\Delta(g,f,h) = D_{\lambda, \alpha}(g,h) - D_{\lambda, \alpha}(g,f) - D_{\lambda, \alpha}(f,h) = O(\varepsilon\nu),
\]
where
\[
\nu = \lambda + \bar{\lambda}\max\left\{\int \delta f^{\alpha}, \int \delta h^{\alpha}\right\}.
\]
\end{thm}
\begin{rem}\label{remspur}
Theorem 3.4 is one of the main theoretical reasons for the behaviour of the minimum BDPD estimators to be observed in the subsequent sections over various choices of tuning parameters and various contaminating distributions. Note that the function
\[
\xi(\lambda) = \frac{1}{\bar{\lambda}}\frac{\lambda + \bar{\lambda}a}{\lambda + \bar{\lambda}b},
\]
for fixed positive values of $a$ and $b$, is necessarily increasing in $\lambda\in[0,1)$ if and only if $a\le b$. This implies that when
\begin{equation}\label{lesscon}
\int \delta h^{\alpha} \le \int fh^{\alpha},
\end{equation}
the error term $T_{\varepsilon,\delta}$ is an increasing function of $\lambda\in[0,1)$ and so the bias of the minimum BDPD estimator is expected to be an increasing function of $\lambda\in[0,1)$; however as the Pythagorean relation does not exist for the DPD, calculations based on $\xi(\lambda)$ are not helpful in theoretically comparing the bias of the minimum DPD estimator. The condition \eqref{lesscon} with $h = f_{\theta}$ and $f = f_{\theta_0}$ for $\theta,\theta_0\in\Theta$ is implied by the usual approximate singularity condition in the robustness literature; the latter condition dictates that the contaminating distribution (in this case represented by $\delta$) in the gross error model is approximately singular with the parametric family so that $\int \delta h^{\alpha}\approx 0$ and $\int fh^{\alpha}$ is relatively large. In the numerical simulations, we will see that this reasoning generally fits in well with the observed behaviour of the estimators except for very small values of $\alpha$; in the latter case the quantities (in Equation \eqref{lesscon}) can be very close and the above observations may not hold.

But when the contaminating distribution $\delta$ is concentrated at (or around) the mode of the model density $h = f_{\theta_0}$, the integral $\int \delta h^{\alpha}$ is often at least as large as $\int fh^{\alpha}$ so that the changing behaviour of $\xi(\lambda)$ over $\lambda$ with $a = \int \delta h^{\alpha}$ and $b = \int fh^{\alpha}$, is less predictable. We will observe in the rest of the paper that the performance of the minimum LDPD estimator (and several other minimum BDPD estimators) may suffer badly in this case. 
\end{rem}
In the following sections, we will consider two kinds of contamination for the parametric model. The first case will correspond to the approximate singularity idea of the gross error model where the contaminating (minor) component is well-separated from the target (major) component. We will refer to this as \emph{outer contamination}; in this case the contaminating values will be \emph{surprising observations} in the sense of \cite{LIND94}. In the second case of contamination, we will choose the contaminating component near the mode of the major component so that these observations are no longer surprising observations but will nevertheless distort the shape of the distribution relative to the parametric model. We will refer to this case as \emph{inner contamination}.     
\section{Spurious Behaviour under Inner Contamination}\label{type0s}
\cite{jhhb01} reported that in case of the exponential distribution, \emph{small outliers} (near the origin) made the minimum LDPD estimator non-robust. These small outliers are essentially what we have described as constituting inner contamination in Section \ref{prop}. Let $\eta_{\theta}(x)$ denote the density of the exponential distribution with mean $\theta$. The non-robustness of the minimum LDPD estimator under small outliers was demonstrated by \cite{jhhb01} in simulation studies and was confirmed by determining the LDPD between the densities $\eta_{\theta}$ and the $0.85\eta_1 + 0.15\delta_{x_0}$ mixture with $x_0 = 0.0001$ where $\delta_{x_0}$ represents the indicator function at $x=x_0$. In Figures \ref{fig21} and \ref{fig22}, we have exhibited the BDPD objective function between the above two densities over $\theta$ for $\alpha = 0.5$ and several values of $\lambda$ in $[0,1]$. 
\begin{figure}[H]\centering
\resizebox{\textwidth}{!}{%
\includegraphics{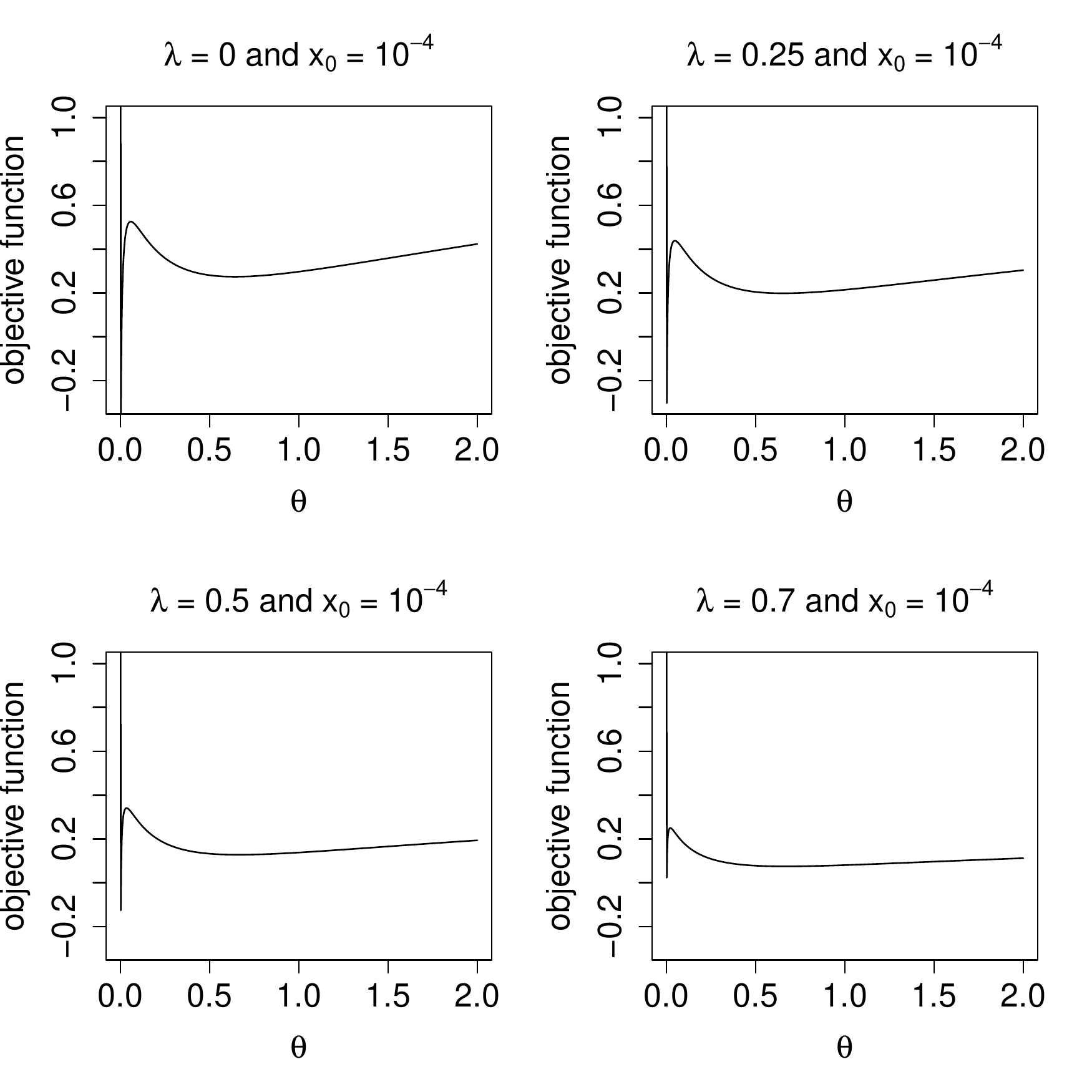}}
\caption{Plots of the BDPD objective function $(\lambda$ $=$ $0.0$,$0.25$,$0.5$,$0.7)$}
\label{fig21}
\end{figure}
\begin{figure}[H]\centering
\resizebox{\textwidth}{!}{%
\includegraphics{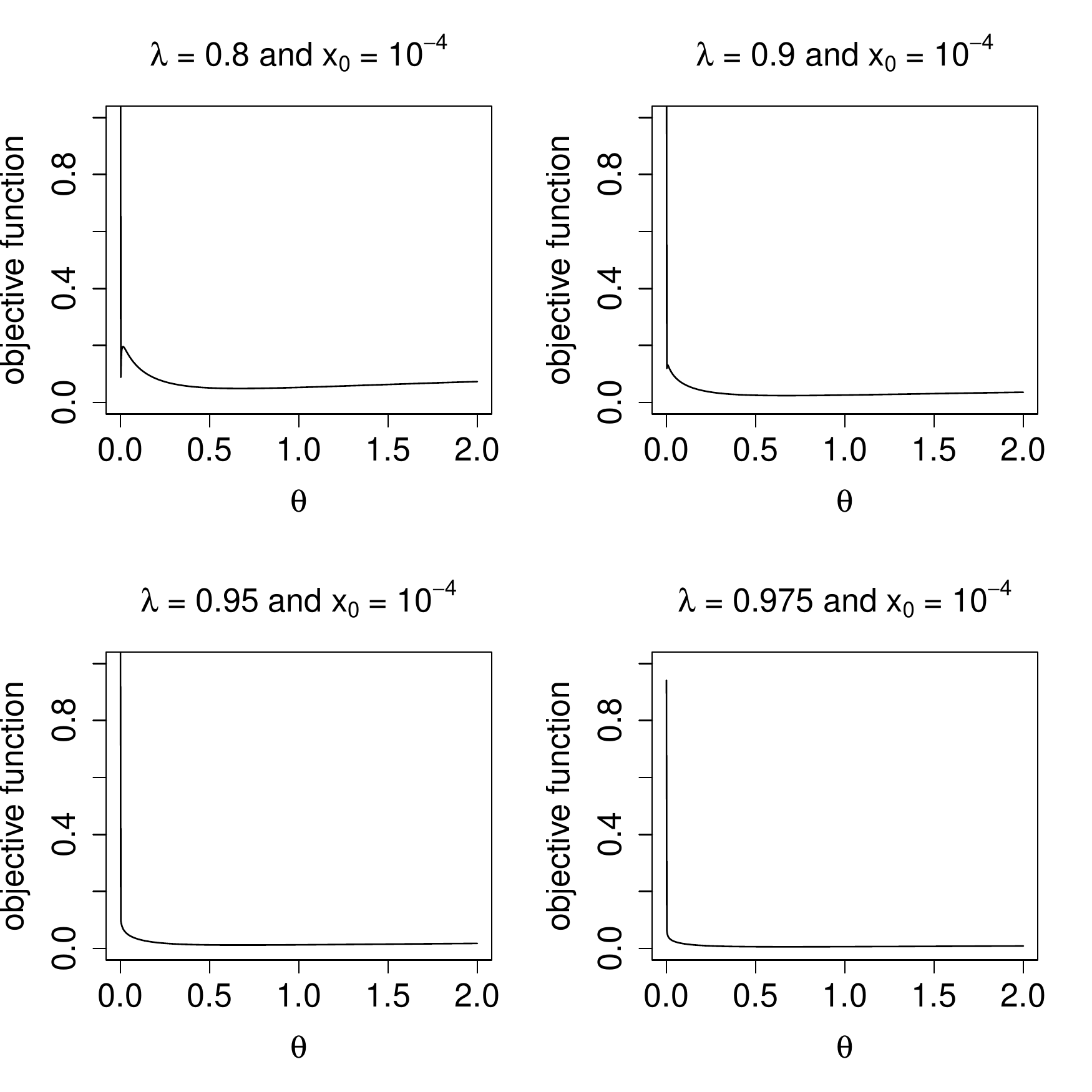}}
\caption{Plots of the BDPD objective function $(\lambda$ $=$ $0.8$,$0.9$,$0.95$,$0.975)$}
\label{fig22}
\end{figure}

It is clear that the global minimizer of the objective function remains stuck at a value very close to zero at least up to $\lambda = 0.7$. It is noteworthy that this spurious minimum may easily be missed by any gradient descent root search algorithm since the minimum is obtained as a needle sharp notch over a very limited region. It might very well be possible (and indeed it is the case in Figures \ref{fig21} and \ref{fig22}) that a local minimizer of the LDPD may serve as a reasonable robust solution even in this case. But the asymptotic results for a sequence of roots of the estimating equation present in the literature do not prescribe the method to \emph{choose} a suitable sequence of roots in case of multiple roots. In any case, it is clear that the global minimizer of the LDPD up to (at least) $\lambda = 0.7$ is a nonsensical value.

The main reason for this phenomenon in the LDPD and the bridge divergences close to it is the behaviour of the $\log$ function at $0$, where it diverges to $-\infty$. If one writes the LDPD between the $\eta_{\theta}$ and $0.85\eta_1 + 0.15\delta_{0.0001}$ densities, we can easily check that there is a sharp drop in the objective function around a value very close to zero because of the above logarithm effect (although at 0, the objective function is positive infinity). This effect slowly smooths out as the $\lambda$ parameter increases since the presence of the additive $\lambda$ term in each of the logarithms eventually forces the argument to be bounded away from zero. The validity of this reasoning is confirmed by examining the behaviour of the LDPD and the BDPD close to it in case of the $N(0,\sigma^2)$ model, where it can be expected that the estimate of $\sigma$ will be driven to zero if there are some outliers near $0$ in the sample. For that matter, this behaviour can also be seen in case of the $N(\mu,\sigma^2)$ model with some outliers near the true mean. Figure \ref{plotnorm} exhibits the LDPD at $\alpha = 0.5$ as a function of $\sigma$ when computed between the densities of the $N(0,\sigma^2)$ and the mixture $0.85N(0,1) + 0.15\delta_{0.001}$. Clearly, there is a reasonable local minimum around $\sigma = 0.8$; but it is beaten hands down by the useless global minimum in the neighbourhood of zero.
\begin{figure}[H]\centering
\resizebox{\textwidth}{!}{%
\includegraphics{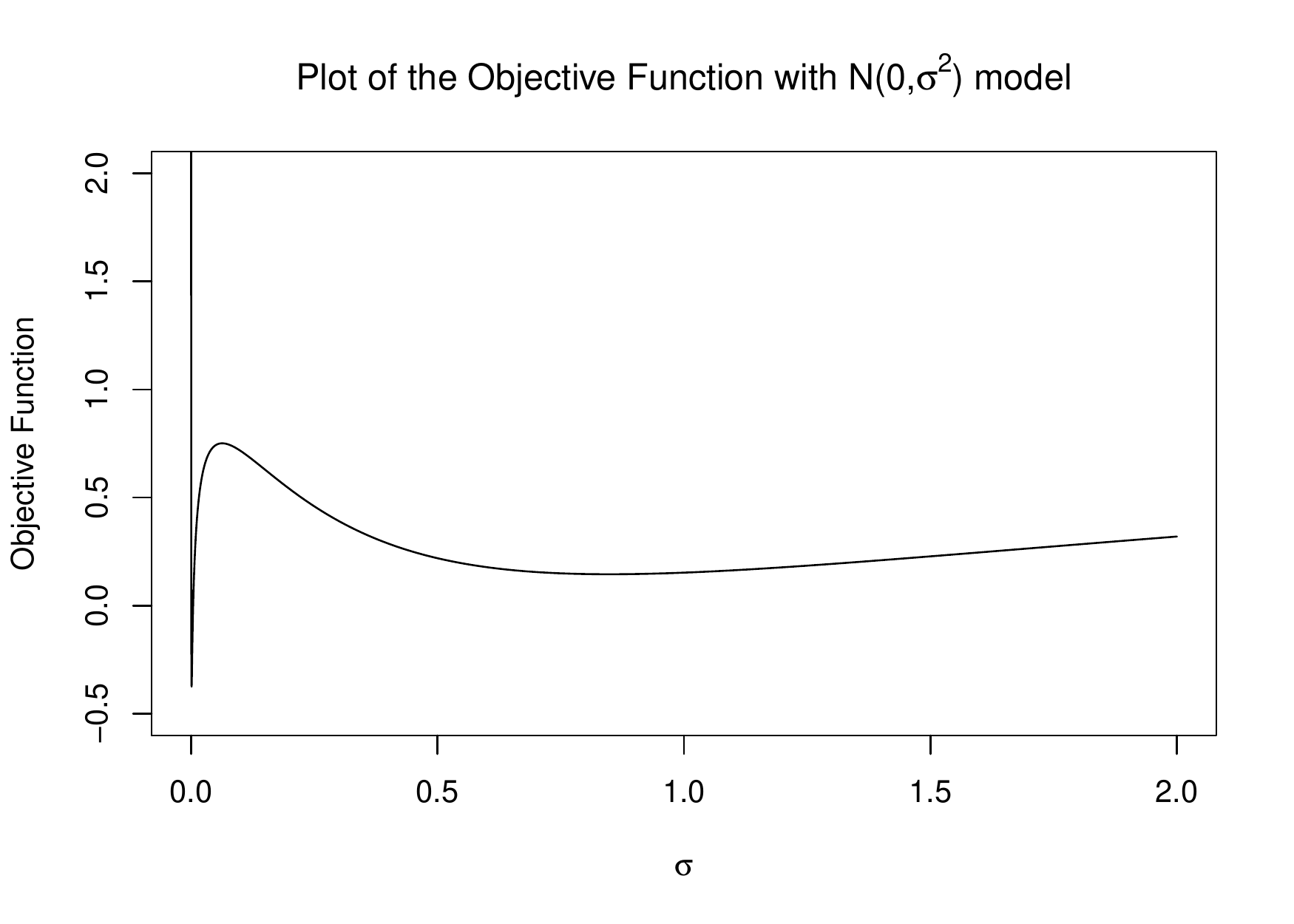}}
\caption{Plot of the LDPD objective function under the $N(0,\sigma^2)$ model.}
\label{plotnorm}
\end{figure}
While the behaviour observed in Figures \ref{fig21}--\ref{plotnorm} represent the patterns in the divergences between the actual densities, such behaviour can also be observed under pure data (although it is relatively rare in scale models). An actual random sample of size 20 from $N(0,1)$ with seed 129 was obtained in R as
\begin{equation*}
\begin{array}{rrrrrr}
-1.120900,& -0.989724,& -1.374697,& -1.355645,&  1.996755,&  0.695870,\\  
 0.771968,& -0.002847,&  1.008549,& -0.990280,&  1.131772,& -0.244929,\\
 1.186625,& -1.671537,& -0.081999,& -1.831365,&  0.358867,&  0.891639,\\
 0.489801,&  0.000010. & & & &
\end{array}
\end{equation*}
The exact value of the last observation is $1.048653\times 10^{-5}$, which forces a spurious behaviour in the LDPD objective function plotted in Figure \ref{plot_sample} for $\alpha = 0.8$. To be precise, if $f_{\sigma}(x)$ denotes the probability density of $N(0,\sigma^2)$ with respect to the Lebesgue measure, then
\[
\int f_{\sigma}^{\alpha + 1}(x) dx = \frac{1}{\sqrt{\alpha + 1}(\sqrt{2\pi}\sigma)^{\alpha}},
\] 
and the LDPD objective function satisfies
\begin{align*}
M_n(\sigma) &
\le \log\left(\frac{1}{\sqrt{\alpha + 1}(\sqrt{2\pi}\sigma)^{\alpha}}\right) - \left(\frac{\alpha + 1}{\alpha}\right)\log\left(\frac{1}{n}f_{\sigma}^{\alpha}(X_{20})\right)\\
&= \log\left(n^{1 + 1/\alpha}\sigma\right) + \frac{1}{2}\log\left(\frac{2\pi}{\alpha + 1}\right) + \left(\frac{\alpha + 1}{2}\right)\frac{X_{20}^2}{\sigma^2}.
\end{align*}
where $X_{20}$ is the last observation which is close to zero and $n = 20$. This inequality confirms our reasoning. As $\sigma$ slides down towards zero, for a while the first term involving the logarithm dominates, pulling the objective function down. However as $\sigma$ gets really close to zero, the third term takes over and there is an extremely sharp rise in the objective function, which generates a minimum with a razor sharp notch. Figure \ref{plot_sample} shows the spurious (global) minimum near zero (at $\sigma= 0.00001309906$ to be exact), although there is a reasonable local minimum at $\sigma = 1.265882$. For this sample, the spurious global minimum phenomenon is observed for all $\alpha \ge 0.500144205$ and all $\lambda\in[0,0.228827308]$. In contrast, the global minimum of the DPD objective function (also plotted in Figure \ref{plot_sample}) for $\alpha = 0.8$ is obtained at $\sigma = 1.20833.$ If the last observation $1.048653\times 10^{-5}$ were removed from this data set, this spurious minimum behaviour of the LDPD objective function disappears. The minimum LDPD estimator now equals $\sigma = 1.321807$ at $\alpha = 0.8$, an entirely reasonable value. (The corresponding minimum DPD estimator is $1.298228$). Thus one single observation can bring about an absolutely drastic change in the minimum LDPD estimator which is against the spirit of stability that robust estimators should have; so far we (or indeed anybody else), have not detected such spurious behaviour in any scenario involving the DPD.

Thus the spurious root issue in case of the LDPD is not a isolated problem limited to the case of the exponential distribution. It is, in fact, a more serious problem in the case of the location-scale model (compared to just the scale model), as we will see in the next section.
\begin{figure}[H]\centering
\resizebox{\textwidth}{!}{%
\includegraphics{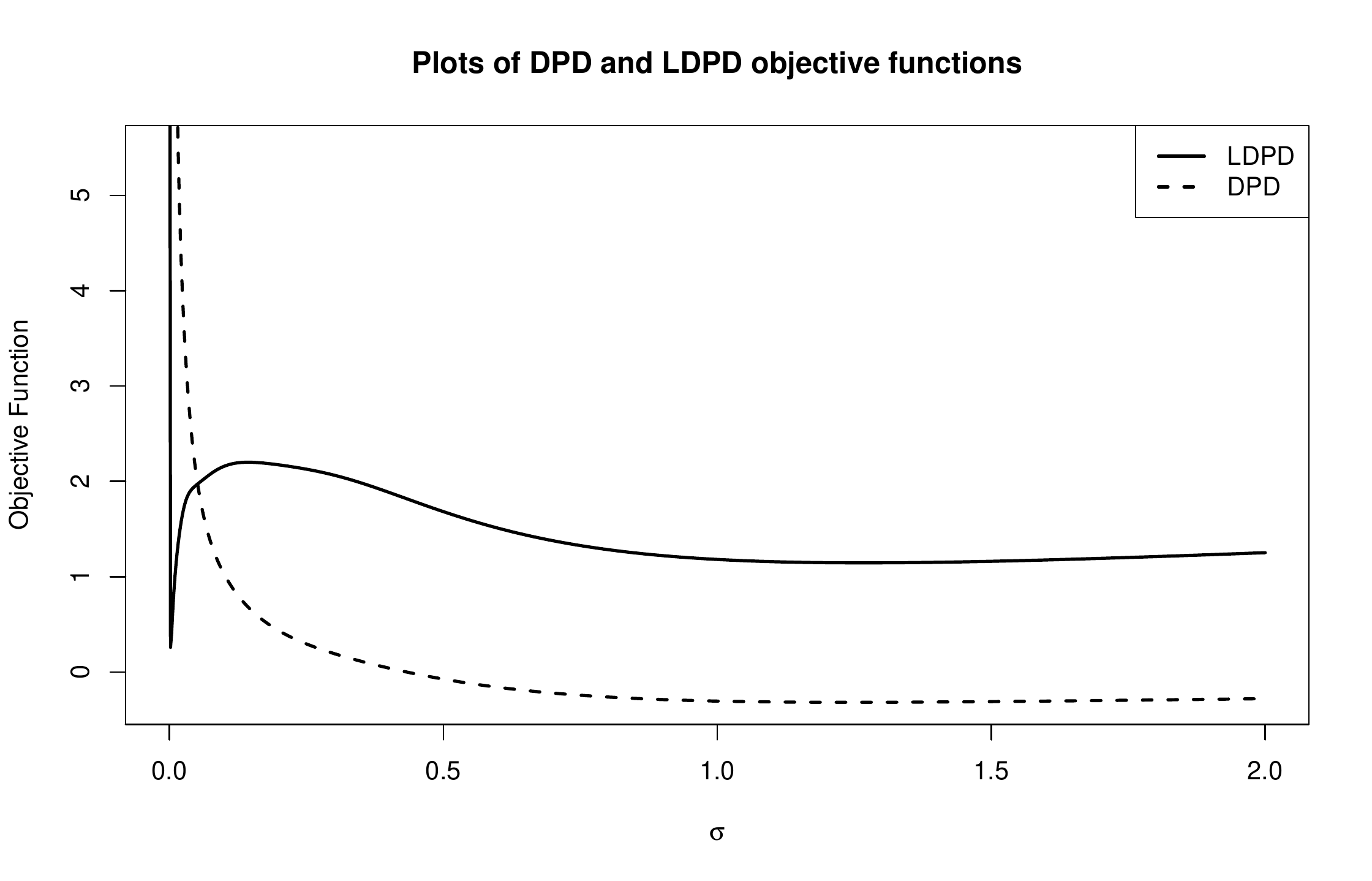}}
\caption{Plots of the sample based DPD and LDPD objective functions under the $N(0,\sigma^2)$ model based on a sample of size 20 from $N(0,1)$.}
\label{plot_sample}
\end{figure}
\section{Unboundedness of the Bridge Divergences}\label{sec:Unbounded}
In the previous section, we explained the reason for observing a spurious minimum with LDPD in the case of inner contamination and argued that it might happen even in case of real data generated from the pure model. Intuitively, this may be explained by the fact that the probability of observing a value near the mode of the majority distribution is not too small at moderate sample sizes. 

Here we will formally prove that the sample version of the BDPD objective function (given on the right hand side of Equation \eqref{eq:BridgeDivergence}) is unbounded below in case the parametric model family is a location-scale family $\mathcal{G}_{f,\Theta}$, where $f$ is a probability density function, $\Theta = \mathbb{R}\times\mathbb{R}^+$, and 
\[
\mathcal{G}_{f,\Theta} = \left\{f_{\theta}(\cdot):\,f_{\theta}(x) = \frac{1}{\sigma}f\left(\frac{x - \mu}{\sigma}\right)\mbox{ for some }\theta := (\mu,\sigma)\in\Theta\right\}.
\]
This result proved in Theorem \ref{thm:UnboundedBelow} implies that one cannot fit a location-scale family using the minimizer of a BDPD, except the DPD.
\begin{thm}\label{thm:UnboundedBelow}
Let $X_1,X_2,\ldots,X_n$ be independent and identically distributed observations from a density $g,$ which is modeled by the densities in the location-scale family of distributions $\mathcal{G}_{f,\Theta}$ with a fixed $f$ satisfying $f(0) > 0.$ Then, for any $0\leq\lambda<1$ and $\alpha > 0,$  
\begin{equation}\label{eq:UnboundedBelow}
\inf_{\theta \in \Theta} \left[\log\left(\lambda + \bl\int f_{\theta}^{1+\alpha}(x)dx\right) - \left(\frac{1+\alpha}{\alpha}\right)\log\left(\lambda + \frac{\bl}{n}\sum_{i=1}^{n} f_{\theta}^{\alpha}(X_i)\right)\right] = -\infty.
\end{equation}
\end{thm}
\begin{proof}
Set the objective function on the left hand side of Equation \eqref{eq:UnboundedBelow} as $M_n(\theta)$. Then for $\theta = (\mu, \sigma),$
\begin{equation*}
M_n(\theta) = \log\left(\lambda + \frac{\bl}{\sigma^{\alpha}} \int f^{1+\alpha}(x)dx\right) -  \frac{1+\alpha}{\alpha}\log\left(\lambda + \frac{\bl}{n \sigma^\alpha}\sum_{i=1}^{n} f^{\alpha}\left(\frac{X_i - \mu}{\sigma}\right)\right). 
\end{equation*}
We will now show that for each $1\le j\le n$, $M_n(X_j, \sigma)$ (that is, $\theta = (X_j, \sigma)$) converges to $-\infty$ as $\sigma\downarrow 0$. Fix $1\le j\le n$. Since $f(x) \ge 0$ for all $x$, we get, by taking $\mu = X_j$,
\[
\left(\frac{1+\alpha}{\alpha}\right)\log\left(\lambda + \frac{\bl}{n \sigma^\alpha}\sum_{i=1}^{n} f^{\alpha}\left(\frac{X_i - \mu}{\sigma}\right)\right)
\geq \left(\frac{1+\alpha}{\alpha}\right)\log \left(\frac{\bl}{n \sigma^\alpha} f^{\alpha}(0)\right).
\]
Substituting this bound in $M_n(\theta) = M_n(X_j, \sigma)$, we obtain
$$M_n(X_j,\sigma) \leq \log \left(n^{\frac{1+\alpha}{\alpha}} \left(\lambda \sigma^{\alpha + 1} + \bl\sigma\int f^{1+\alpha} \right)\right) - C_f,$$
where
\[
C_f = (1+\alpha)\log\left(\bl^{\frac{1}{\alpha}} f(0)\right) .
\]
Letting $\sigma$ tend to zero implies $M_n(X_j, \sigma)\to -\infty$ proving Equation \eqref{eq:UnboundedBelow}.
\end{proof}
\begin{rem}
From the proof, it is seen that the global minimizer of the bridge divergence in case of any sample is at the extreme point ($\sigma = 0$) and the global minimum over all $(\mu, \sigma)$ combinations, is attained for atleast $n$ points namely $(X_j, 0),\,1 \le j\le n$. This is similar in spirit to the classical example of likelihood inference for Gaussian mixture modelling as explained in Example 2.4.5 of \cite{BD15}. Note that the theorem does not cover the case of DPD which is given below.  
\end{rem}
\begin{figure}[H]\centering  
\resizebox{\textwidth}{!}{%
\includegraphics{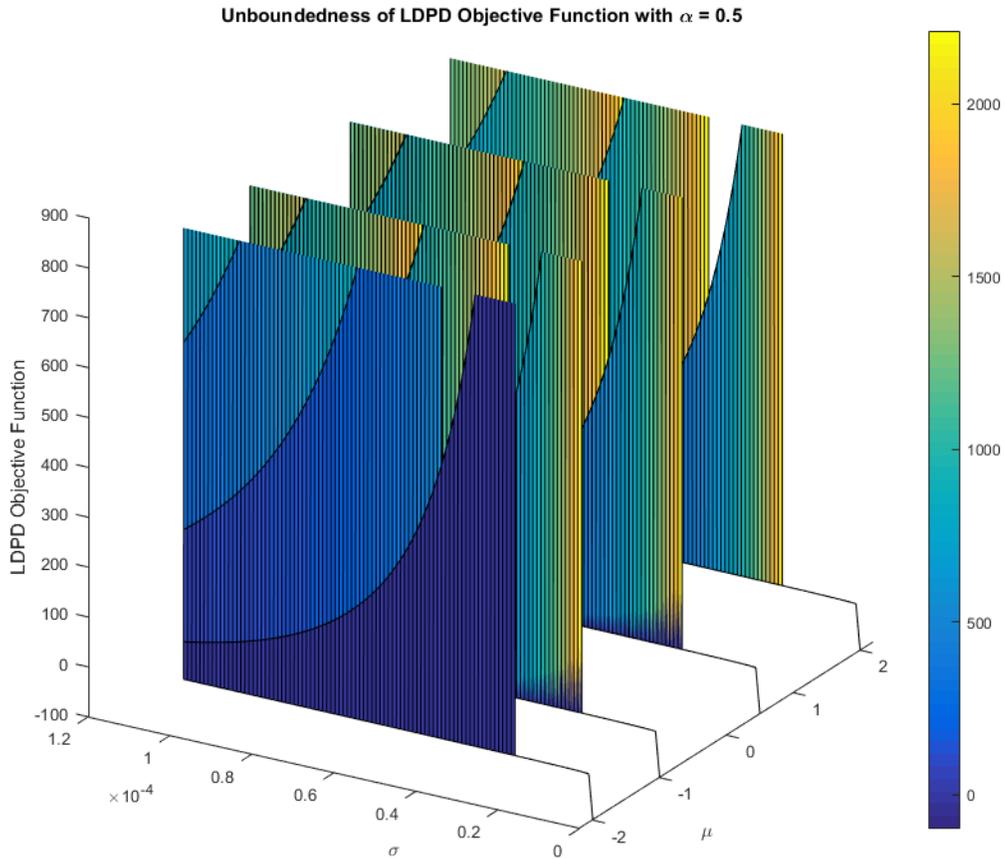}}\vspace{-24mm}
\caption{LDPD objective function for the $N(\mu,\sigma^2)$ model with an artificial dataset of four observations.}
\label{thd}
\end{figure}

In order to illustrate the phenomenon discussed above, we plot the LDPD objective with $\alpha = 0.5$ in Figure \ref{thd}, based on an artificial sample of size 4, where the sample observations are $-2,-1,0.5,2.$ The underlying model family is assumed to be $\{N(\mu,\sigma^2): \mu \in \mathbb{R},\sigma > 0\}$. We observe that the LDPD objective drops down sharply to $-\infty$ as $\sigma$ approaches $0$, when $\mu$ is either of the four sample points. See Section S.2 of supplementary materials for more details.
\begin{thm}\label{thm:DPDUnboundedAbove}
Let $X_1,X_2,\ldots,X_n$ be independent and identically distributed observations from a density $g,$ which is modeled by the densities of the location-scale family $\mathcal{G}_{f,\Theta}$ with a fixed $f$. Suppose that $\lim_{|x| \rightarrow \infty} f(x) = 0.$  Then, for any $\alpha > 0,$ the following holds for all $n$ if $\mu$ does not belong to the set $\{X_1,...,X_n\}$, and for all $n > \frac{(1+\alpha)f^{\alpha}(0)}{\alpha \int f^{1+\alpha}}$ if $\mu$ is equal to some $X_i$ in the above set:  
\begin{equation}\label{eq:UnboundedAbove}
\lim \limits_{\sigma \rightarrow 0} \left(\int f_{(\mu,\sigma)}^{1+\alpha}(x)dx - \left(\frac{1+\alpha}{\alpha}\right)\frac{1}{n}\sum_{i=1}^{n} f_{(\mu,\sigma)}^{\alpha}(X_i)\right) = \infty.
\end{equation}
\end{thm}
\begin{proof}
First, note that the quantity on the left hand side of \eqref{eq:UnboundedAbove} whose limit needs to be evaluated as $\sigma \rightarrow 0$, can be written as
\begin{equation}\label{eq:sclconv}
\sigma^{-\alpha}\left(\int f^{1+\alpha} - \left(\frac{1+\alpha}{\alpha n}\right)\sum_{i=1}^{n} f^{\alpha}\left(\frac{X_i-\mu}{\sigma}\right)\right).
\end{equation}
Observe that, if $\mu$ does not belong to the set $\{X_1,...,X_n\}$, then $$\lim \limits_{\sigma \rightarrow 0} f^{\alpha} \left(\frac{X_i - \mu}{\sigma}\right) = 0\quad\mbox{forall}\quad 1 \leq i \leq n.$$
Consequently, \eqref{eq:sclconv} goes to $\infty$ as $\sigma \rightarrow 0.$ On the other hand, if $\mu$ equals $X_j$ for some $1 \leq j \leq n,$ then \eqref{eq:sclconv} can be written as:
\begin{equation}\label{eq:exc2}
\sigma^{-\alpha}\left(\int f^{1+\alpha} - \left(\frac{1+\alpha}{\alpha n}\right) f^\alpha (0)- \left(\frac{1+\alpha}{\alpha n}\right)\sum_{i\neq j} f^{\alpha}\left(\frac{X_i-\mu}{\sigma}\right)\right)
\end{equation}
It is now easy to see that for $n > \frac{(1+\alpha)f^{\alpha}(0)}{\alpha \int f^{1+\alpha}}$, \eqref{eq:exc2} goes to $\infty$ as  $\sigma \rightarrow 0.$ This completes the proof. 
\end{proof}
\begin{rem}
Theorem \ref{thm:DPDUnboundedAbove} shows that the DPD objective function for the location-scale family (for $\alpha > 0$) diverges to $\infty$ as $\sigma \rightarrow 0,$ for all $n,$ if $\mu$ is different from all the observations. On the other hand, if $\mu$ equals one of the observations, the situation differs in the sense that there exists $a_f$ such that if $n > a_f$, the DPD objective function diverges to $\infty$ and if $n < a_f$, it diverges to $-\infty$. The case of $n = a_f$ (when $a_f$ is a natural number) depends on the rate of convergence of $f(x)$ to zero as $|x|\to\infty$.
\end{rem}
\begin{rem}
Theorem \ref{thm:UnboundedBelow}, in particular, covers the case of $N(\mu, \sigma^2)$ parametric family and out of entire BDPD family the DPD alone stands out as a practically valid method of inference. \cite{fe08} provided an iterative algorithm to obtain the \textit{minimizer} of the LDPD in case the parametric family is one of the exponential families, applicable to $N(\mu , \sigma^2)$ family. However, while their iterative algorithm is monotone decreasing, it does not in general guarantee convergence to the global minimizer. In fact, the small bias property shown for the scale parameter in their numerical study (Table 2, \cite{fe08}) suggests, in view of our discussion in Section \ref{type0s}, that their algorithm mostly converges to a local minimizer, atleast in the normal family. This is also very similar to the case of likelihood inference for Gaussian mixture modelling wherein one uses the EM algorithm to get hold of the useful local minimizer of the log-likelihood. This fact is also consistent with the asymptotic normality result presented in Section 5 of \cite{fe08}. The asymptotic normality holds for a sequence of roots of the estimating equation while the global minimizer does not exist (in the interior of the parameter space). Note that in this case the global minimizer does not appear as a root of the estimating equation owing to the fact that any global minimum is attained on the boundary of the parameter space.
\end{rem}
\begin{rem}\label{rem:ss}
The result as stated in Theorem \ref{thm:UnboundedBelow} does not apply to a scale family with a fixed location (say, $0$). But if one of the observations is exactly $0$, then the proof as presented above can be employed to get the same conclusion. More precisely, we have, (taking $\sigma = X_{(1)}$),
\[
\inf_{\sigma > 0}\left[\log\int f_{\sigma}^{1+\alpha} - \left(1 + \frac{1}{\alpha}\right)\log\left(\frac{1}{n}\sum_{i=1}^n f_{\sigma}^{\alpha}(X_i)\right)\right] \le \log(n^{1 + 1/\alpha}X_{(1)}) + D_{f},
\]
where $D_{f} = \log\int f^{1+\alpha} - \left(1+{\alpha}\right)\log f(1)$ under the assumption $f(1) > 0$. Here $X_{(1)}$ denotes the first order statistic in the sample $X_1, \ldots, X_n$ and $f_{\sigma}(x) = f(x/\sigma)/\sigma$ for a density $f$ on $\mathbb{R}^{+}$. This indicates that if we have inner contamination (that is, contamination near the mode $0$) so that $n^{1 + 1/\alpha}X_{(1)}$ is small enough, then there is a possibility of a spurious global minimum near zero. Observe that the true distribution is a mixture described as
\[
X|Z = 0 \sim f_{\sigma_0}\quad\mbox{and}\quad X|Z = 1 \sim \delta_0,
\]
where $Z\sim \mbox{Bernoulli}(p)$ and $\delta_0$ represents a point mass at $0$. This represents a simple case inner contamination model with contamination level $p$. Then in a sample of size $n$, we have for any fixed $\varepsilon > 0,$
\[
\mathbb{P}(X_{(1)} > n^{-1-1/\alpha}\varepsilon) = \left(1 - p - (1 - p)F\left(n^{-1-1/\alpha}\varepsilon/\sigma_0\right)\right)^n \le (1 - p)^n,
\]
implying convergence with probability one of $n^{1 + 1/\alpha}X_{(1)}$ to zero as $n\to\infty$ for any fixed $p> 0$. Here $F$ represents the cumulative distribution function of $X$. In a given sample though, this spurious minimum may disappear (or may not be significantly pronounced) depending on the proportion of contamination. In this case, the BDPD do not behave so badly as in the location-scale family. As already observed, the divergences close to the DPD will lead to reasonable estimators.
\end{rem}
\section{Towards the Desired Estimator}\label{sec:Chain_Algorithm}
As we have seen, the analysis using the LDPD can be beset with different kinds of problems, both theoretical and computational. In this section we add to the discussion of the possible flaws in the current state of the analysis based on the LDPD, eventually giving a recipe for consolidating the existing knowledge and overcoming the present difficulties so that one can arrive at a best compromise. We do not exactly refute the findings of \cite{fe08} and \cite{fuji13} as we find a substantial part of this research to be valuable and useful. However, there are too many loose ends that remain unaccounted for which limit the practical applications of the method without further consolidation. 

\cite{fe08} and \cite{fuji13} define the latent bias of an estimator as the difference between the target parameter and the limit of the estimator. To describe the basic flaw with the minimized LDPD estimator or in general the normalized estimating equations, we rework the arguments of \cite{fuji13} which gives the hint of a very small latent bias under heavy contamination for the minimum LDPD estimator. Let $f(x) = f_{\theta^*}(x)$ be the target density within our parametric family and $\delta(x)$ be the contamination density so that the data generating density $g$ is given by
\[
g(x) = (1-\varepsilon)f(x) + \varepsilon \delta(x).
\]
Let $\ell(x;\theta) = \log f_{\theta}(x)$ and $u_{\theta}(x) = \nabla\ell(x;\theta)$ be the usual log-likelihood and the score, respectively. A general normalized estimating equation is given by
\begin{equation}\label{eq:NormEstEq}
\frac{\mathbb{E}_g[\xi(\ell(X;\theta))u_{\theta}(X)]}{\mathbb{E}_g[\xi(\ell(X;\theta))]} = \frac{\mathbb{E}_{f_{\theta}}[\xi(\ell(X;\theta))u_{\theta}(X)]}{\mathbb{E}_{f_{\theta}}[\xi(\ell(X;\theta))]},
\end{equation}
where $\xi(\cdot)$ is a non-negative weight function satisfying $\xi(a)\to0$ as $a\to-\infty$ and $\mathbb{E}_h[\cdot]$ represents the expectation with respect to the density $h$. Assume that
\begin{equation}\label{eq:ApproxOutlier}
\mathbb{E}_{\delta}[\xi(\ell(X;\theta))] \approx 0,\quad\mbox{and}\quad \mathbb{E}_{\delta}[\xi(\ell(X;\theta))u_{\theta}(X)]\approx 0,
\end{equation}
in a neighborhood of $\theta = \theta^*$, which is usually satisfied under the classical outlier model formulation (or, as we call it, outer contamination). Substituting the form of the density $g$, we get
\[
\frac{(1-\varepsilon)\mathbb{E}_{f_{\theta^*}}[\xi(\ell(X;\theta))u_{\theta}(X)] + \varepsilon\mathbb{E}_{\delta}[\xi(\ell(X;\theta))u_{\theta}(X)]}{(1-\varepsilon)\mathbb{E}_{f_{\theta^*}}[\xi(\ell(X;\theta))] + \varepsilon\mathbb{E}_{\delta}[\xi(\ell(X;\theta))]} = \frac{\mathbb{E}_{f_{\theta}}[\xi(\ell(X;\theta))u_{\theta}(X)]}{\mathbb{E}_{f_{\theta}}[\xi(\ell(X;\theta))]}.
\]
Whenever the approximations in \eqref{eq:ApproxOutlier} hold, we get the following approximate equality
\[
\frac{\mathbb{E}_{f_{\theta^*}}[\xi(\ell(X;\theta))u_{\theta}(X)]}{\mathbb{E}_{f_{\theta^*}}[\xi(\ell(X;\theta))]} \approx \frac{\mathbb{E}_{f_{\theta}}[\xi(\ell(X;\theta))u_{\theta}(X)]}{\mathbb{E}_{f_{\theta}}[\xi(\ell(X;\theta))]}.
\]
Thus, $\theta^*$ is an approximate solution of the normalized estimating equation \eqref{eq:NormEstEq}. For concreteness, consider the case where $f(x) = f_{0}(x)$ (i.e., $\theta^* = 0$), where $f_{\mu}$ denotes the density of the normal distribution with mean $\mu$ and variance $1$ and $\delta(x) = f_{10}(x)$. Take $\theta = 10$ in the normalized estimating equation \eqref{eq:NormEstEq}. It is easy to see that
\[
\mathbb{E}_{f_0}[\xi(\ell(X;\theta))]\approx 0, \quad\mbox{and}\quad \mathbb{E}_{f_0}\left[\xi(\ell(X;\theta))u_{\theta}(X)\right]\approx 0,
\]
in a neighborhood of $\theta = 10$. Therefore, as above, the equality
\[
\frac{\mathbb{E}_{f_{10}}[\xi(\ell(X;\theta))u_{\theta}(X)]}{\mathbb{E}_{f_{10}}[\xi(\ell(X;\theta))]} = \frac{\mathbb{E}_{f_{\theta}}[\xi(\ell(X;\theta))u_{\theta}(X)]}{\mathbb{E}_{f_{\theta}}[\xi(\ell(X;\theta))]},
\]
holds approximately. Thus, not only $\theta^* = 0$ but $\theta = 10$ is also an approximate root of this estimating equation. The latent bias argument for normalizing estimating equations would be one sided if one were to use it only for $\theta^*$, and not for the root corresponding to the contaminating component. One probably can have an estimator with small latent bias in this case, but one would have to appropriately choose between the multiple roots to ensure this small bias.

To validate our argument, we generated 20 observations randomly from the mixture distribution $0.9 N(10,1) + 0.1 N(0,1)$. We present the entire sample for the future reproducibility of our results.
\begin{equation*}
\begin{array}{rrrrrr}
 10.73402487, & 10.07316778, & 12.06112801, &  8.52976810, &  9.66834408,\\
 10.69923160, & 10.65200828, &  9.87964160, &  8.76851845, & 10.91329572,\\ 
  8.12296270, &  8.57642728, & 8.52115355, & 12.42173904, & 10.45406538,\\
 9.86883487, & -0.89282090, & -0.06273843,& 1.34072270, & -0.65001187.
\end{array}
\end{equation*}
The parametric family is taken to be $\mathcal{F} = \{N(\mu,1): \mu \in \mathbb{R}\}.$ The LDPD objective function in \eqref{eq:LDPDSampleObj} with $\alpha = 0.5$ is plotted against $\theta$ in Figure \ref{tworoots}. It is seen to have a (local) minimum near $0$ and a (possibly global) minimum around $10$, which in this case is the mean of the larger component of the mixture. Note that while there are at least two obvious roots of the corresponding estimating equation, a clear root selection strategy is unavailable. While it so happens that in this case the divergence is well behaved, in many standard cases it is not (as we have seen in Sections \ref{type0s} and \ref{sec:Unbounded}).
\begin{figure}[H]\centering
\resizebox{\textwidth}{0.5\textheight}{%
\includegraphics{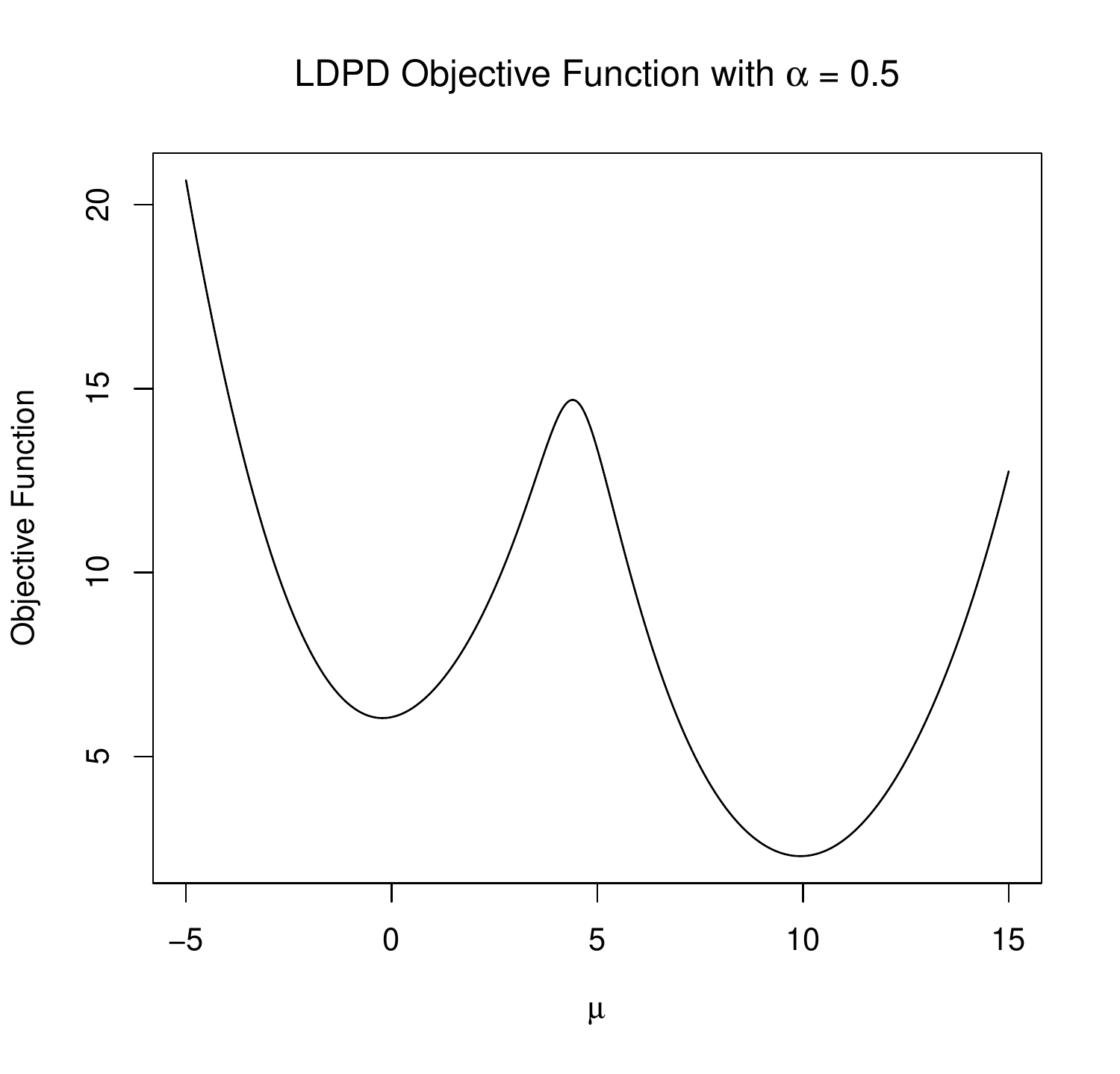}}
\caption{LDPD Objective Function for the $N(\mu,1)$ model under normal mixture data.}
\label{tworoots}
\end{figure}

 We are now in a position to give an analysis of the minimum LDPD estimator in particular (and the minimum BDPD estimator in general) by consolidating the different bits of results that we have presented so far in this paper. In Section \ref{prop} we have proved the consistency of the minimum BDPD estimator. What it means in layman's terms is that the global minimizer of the empirical LDPD objective function converges to the global minimizer of the theoretical objective function obtained by replacing the empirical distribution function by the true one. However, the usefulness of these consistency results are immediately challenged by our findings in Sections \ref{type0s} and \ref{sec:Unbounded} as well as the previous part of this section. First of all, it is quite possible, even in case of pure data, that the global minimizer of the empirical objective function is a value which is entirely useless to the statistician for all practical purposes. Yet, the procedure might actually have a very reasonable local minimizer close to our target. This is what we have observed in Figure \ref{plot_sample}, for example, \cite{fe08} or \cite{fuji13} -- or, indeed, anybody else for that matter -- do not give any root selection criteria in such cases; and in the absence of any such criteria, one would be hard put to justify the choice of a local minimizer, when a perfectly legitimate (although statistically useless) global minimizer exists.  Our results in Section \ref{sec:Unbounded} show that the global minimizers for the LDPD would always be at the boundary of the parameter space and will be of very little practical value to us under location-scale models. However, in almost every such case there exists -- as we have seen in our simulations -- reasonable roots of the estimating equations, representing appropriate local minima of the objective function. The bias and mean square error reported in Figures 2-5 of \cite{fuji13} correspond to the root generated by this reasonable local (but not global) minimum. 
 
 Secondly, an equally serious problem is the extreme shift in the target parameter itself under inner contamination when LDPD is the divergence of choice. As we have seen in Figure \ref{plotnorm}, a small proportion of inner contamination is enough to shift the target parameter (the global minimizer of the divergence between the theoretical densities) to such an extreme degree that the new target does not characterize the original major component in any way at all. Yet, once again, the divergence has a local minimizer which reasonably characterizes the the major component. Once again, however, there is no conceivable reason for considering a local minimizer as the target in the presence of a global minimizer (which, unfortunately, does not characterize the component of interest). 
 
 It is clear, therefore, that usual inference based on the LDPD alone can lead the inference astray, if one proceeds with the global minimizer.
 Other reasonable roots do exist, but root selection strategies do not. Existing literature claims the existence of ``a root'' which has the desired properties, without any recipe for arriving at that root. Such inadequacy exists in case of many other members of the BDPD class. The only member within this class which are entirely free -- as far as it is observed in our explorations -- from the anomalies of spurious roots and ill-defined targets is the DPD.  In every model and every case that we have looked at, the DPD has a well defined minimizer that would reasonably represent, in theory, a properly described target, and would generate a suitable estimator for the same target under randomly generated data. Acknowledging that the LDPD does have certain roots which are superior in dealing with the bias under heavy contamination, we propose to utilize the global minimizer of the DPD to arrive at the desired root of the LDPD. The next subsection describes the chain algorithm which we propose for this purpose.   
 \subsection{The Chain Algorithm}\label{chain}
 Under repeated simulations we have observed, under many standard parametric models, that as far as the global minimizers of the BDPDs are considered over all $\lambda \in [0, 1]$ given a specific value $\alpha$, the minimum DPD estimator is almost always the best in terms of latent bias (and not the minimum LDPD estimator). At the same time we acknowledge that the minimum LDPD estimating equations (or, more generally, the minimum BDPD estimating equations) may have roots which represent some local minimum of the divergences which might improve upon the performance of the minimum DPD estimator in terms of latent bias. As indicated by the heuristics, and as we have seen in our simulations, the LDPD (more generally the BDPD) has a local minimizer that is close to the true $\theta^*$. This we believe is at least partially due to the Pythagorean relation that we presented in Section \ref{prop}. Since the members of the BDPD family form a literal bridge between the DPD and LDPD, we have an opportunity for getting at a root of the LDPD estimating equations that is consistent and asymptotically normal. A completely rigorous proof is unavailable at this time, but we provide a strong heuristic argument and extensive simulations to support this claim.

 From the form of the bridge divergence, for a fixed $\alpha$ and $\theta$, we obtain the limit
 \begin{equation}\label{eq:BDPDClose}
 \rho^{(\alpha, \lambda_k)}(g,f_{\theta}) \to \rho^{(\alpha, \lambda_0)}(g,f_{\theta})\quad\mbox{as}\quad \lambda_k\to\lambda_0.
 \end{equation}
 Heuristically this indicates that $\rho^{(\alpha, \lambda)}(g, f_{\theta})$ has at least a local minimizer close to $\hat{\theta}^{\alpha}_{n1}$ (defined in Equation \eqref{eq:MDPDE}) if $\lambda$ is close to one. For example, one can take $\lambda = 1 - n^{-1/2}$ which is sufficiently close to $1$ for this root to be reasonably good. We shall now present a chain algorithm for getting hold of a good root of the LDPD  estimating equation (indeed, all bridge estimating equations). The chain algorithm proceeds in the following steps: Fix $\alpha\in[0,1].$
 \begin{enumerate}
 	\item First choose a sequence $\{\lambda_i\}$ satisfying
 	\[
 	1 = \lambda_K > \lambda_{K-1} > \cdots > \lambda_2 > \lambda_1 > \lambda_0 = 0\quad\mbox{and}\quad \max_{1\le i\le K}|\lambda_i - \lambda_{i-1}| \le r_n,
 	\]
 	with $r_n \to 0$ at some rate. Define a function $\lambda \mapsto \hat{\theta}^{\alpha}(\lambda).$
 	\item Solve the DPD problem completely, that is, find the global minimizer $\hat{\theta}^{\alpha}_{n1}$ of the sample estimate of the DPD. Set $\hat{\theta}(\lambda_K) = \hat{\theta}^{\alpha}_{n1}$. Note that $\lambda = 1$ in the bridge divergence expression in Equation \eqref{eq:BridgeDivergence} corresponds to the DPD (in the limit).
 	\item For $i = K-1, K-2, \ldots, 0$, find the local minimizer of the sample estimate of the bridge divergence with parameters $(\alpha, \lambda_{i})$ that is closest to $\hat{\theta}(\lambda_{i+1}).$ Set this closest local minimizer as $\hat{\theta}(\lambda_i)$.
 	\item Return $\{\hat{\theta}(\lambda_i):\,0\le i\le K\}$.
 \end{enumerate}
 Step 3 above can be solved by using any off-the-shelf algorithm for minimization with $\hat{\theta}(\lambda_{i+1})$ as a starting point. To draw analogue with existing algorithms of this type, we note that the LASSO that has taken over the literature in the past few years is solved by using a chain/path algorithm as above; there the issue is about getting fast algorithm not distinguishing local and global minimizers (it is a convex problem).
 \begin{conj}
 	In any model under the assumptions that guarantee the asymptotic normality of $\hat{\theta}_{n1}^{\alpha}$, $\hat{\theta}^{\alpha}(\lambda_0)$ is the consistent root of the LDPD estimating equation that is closest to $\theta^*$. And this is the solution claimed by \cite{fe08} to have a small latent bias.
 \end{conj}
 The simulations using this chain algorithm are very encouraging and we think that a proof of the above conjecture can be obtained by an application of the Taylor series and using \eqref{eq:BDPDClose} with an appropriate choice of rate of convergence of $\lambda_k - \lambda_0$ to zero. In many \emph{nicely behaved} parametric densities, simply choosing the root of LDPD estimating equation closest to the minimum DPD might be good enough. However, in general cases, the difference between LDPD and DPD estimating equations can be drastic for this simple trick to work (at least for a theoretical guarantee).

 In the following section we will frequently use the term ``minimum bridge divergence estimator". It is to be understood that this will refer to the local minimizer obtained by using the chain algorithm which uses the global minimizer of the DPD as the starting value. In addition, we will also frequently use the term``root of the estimating equation". In this case it will be understood that the correspondence is with a local minimizer, and not an intervening local maximizer (or saddle point), which may also produce a legitimate root of the estimating equation. 
\section{Simulation Study}\label{sec:Simulations}
In this section, we apply the chain algorithm discussed in Section \ref{sec:Chain_Algorithm} on simulated data. We consider the exponential and the normal scale families, the former to illustrate the case of outer contamination, and the latter to illustrate the case of inner contamination. The estimators are computed from 1000 replications. The bias and mean squared error over these $1000$ replications are calculated against the target component of the underlying majority distribution. 

The chain algorithm is applied for each pair $(\alpha,\epsilon)$, where the robustness parameter $\alpha$ runs from $0$ to $1$ in steps of $0.2$ and the contamination level $\epsilon$ is $0$, $0.05$ or $0.2$. For each $(\alpha, \epsilon)$ combination, the bridge parameter $\lambda$ runs from $1$ to $0$ in steps of $0.1$ as the chain algorithm proceeds. Since the parameter $\alpha$ is somewhat better understood in the literature and our introduced parameter $\lambda$ needs to be analyzed more deeply, the latter is varied through finer steps than the former.

\subsection{Exponential Scale Model}
Here the model is the class of exponential distributions with mean $\sigma$ over $\sigma\in(0,\infty)$. Data are simulated from the mixture $(1-\epsilon)$Exp$(1)$ + $\epsilon$Unif$(6-10^{-4},6+10^{-4})$, where the first component is exponential with rate 1, and the second is uniform over the indicated range. The contamination level $\epsilon$ is taken to be $0, 0.05$ and $0.2$. For step 2 of the chain algorithm which involves solving the DPD problem completely, $25$ starting points are drawn uniformly from the interval $[0,0.1]$ and the remaining $75$ uniformly from the interval $(0.1,10]$. 
\subsection{Normal Scale Model}
Here the model is the class of normal distributions with mean $5$ and variance $\sigma^2$ over $\sigma\in(0,\infty)$. Data are simulated from the mixture $(1-\epsilon)N(5,1)$ + $\epsilon$Unif$(5-10^{-5},5+10^{-5})$, where the first and the second components are normal and uniform respectively, with parameters as indicated. The contamination level $\epsilon$ is taken to be $0, 0.05$ and $0.2$. For step 2 of the chain algorithm which involves solving the DPD problem completely, $25$ starting points are drawn uniformly from the interval $[0,0.1]$ and the remaining $75$ uniformly from the interval $(0.1,10]$.

For the sake of brevity, we will only include the graphs of the mean-squared errors (scaled by sample size $n = 100$) of the minimum bridge divergence estimators in the exponential scale model case, for contamination levels $0.05$ and $0.2$ (in Figure \ref{mseepzeropttwol}, respectively) in the main article. Tables presenting exact values of the bias and MSE of the minimum bridge divergence estimators for both the exponential and normal scale families, for all the three contamination levels $0, 0.05$ and $0.2$ are given in the online supplement. 

It is observed that for the $5\%$ contamination level, the MSE of the minimum bridge divergence estimators decreases as the chain algorithm proceeds ($\lambda$ goes from 1 to 0) for $\alpha$ upto $0.6$, and increases as the chain algorithm proceeds for $\alpha = 0.8$ and $1$. On the other hand, for the $20\%$ contamination level, the MSE of the minimum bridge divergence estimators decreases as the chain algorithm proceeds from $\lambda = 1$ to $\lambda = 0$ for all values of $\alpha$. As a function of $\alpha$ (with $\lambda$ held fixed), for $5\%$ contamination, the MSE decreases roughly upto $\alpha = 0.6$ and then increases with $\alpha$, whereas for $20\%$ contamination, a completely decreasing trend of the MSE is observed as $\alpha$ increases. 

Note that the values reported here are not, except for the DPD ($\lambda = 1$ case) case, based on the global minimizers of the divergences; rather they are the chain algorithm solutions. What this shows is that for heavy contamination (at the level of 20\%) the LDPD solution (not necessarily the global minimizer) as obtained from the chain algorithm does seem to dominate the other BDPD solutions, and the mean square errors decrease uniformly in the direction $1 \rightarrow 0$ over $\lambda$ for each $\alpha$. This provides partial confirmation of the results of \cite{fe08} and \cite{fuji13}. However the results are now stronger and more useful in that we identify the root for which it works; it does not work for just any root, and certainly not for the global minimizer. This is so in spite of the fact that the contamination considered here is an outer contamination. Any exponential distribution has a mode at zero, and randomly throwing up values close to zero, leading to spurious global minimizers, is not entirely uncommon. 

A similar situation is observed in the normal distribution case. Here the contamination is an inner contamination, which leads to spurious global minimizers more often for the BDPDs. However the chain algorithm guides the process to a sensible root (that is also a local minimizer) in each case. 

\begin{figure}
\centering
\begin{minipage}{.5\textwidth}
  \centering
\resizebox{\textwidth}{!}{%
	\includegraphics{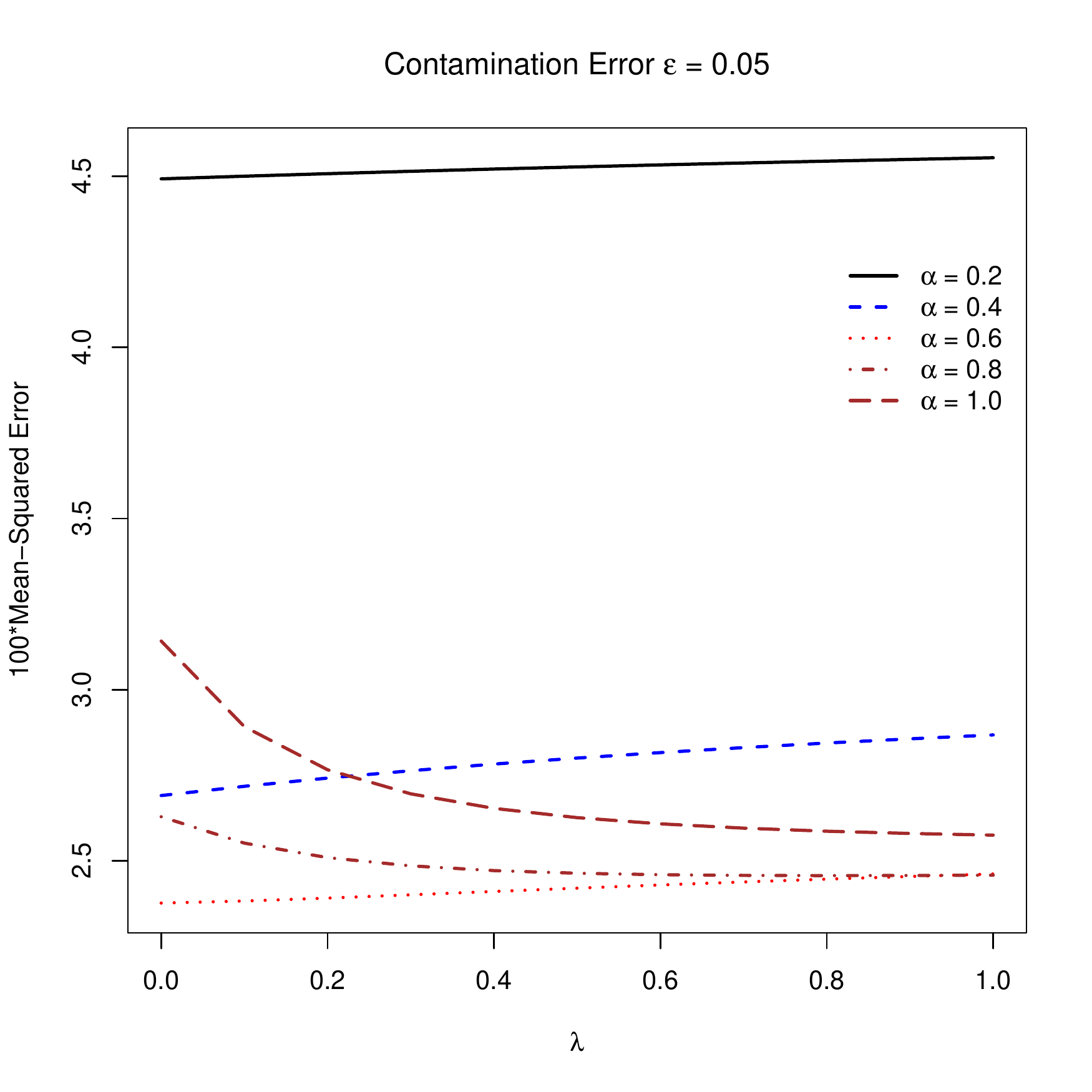}}
\end{minipage}%
\begin{minipage}{.5\textwidth}
  \centering
\resizebox{\textwidth}{!}{%
	\includegraphics{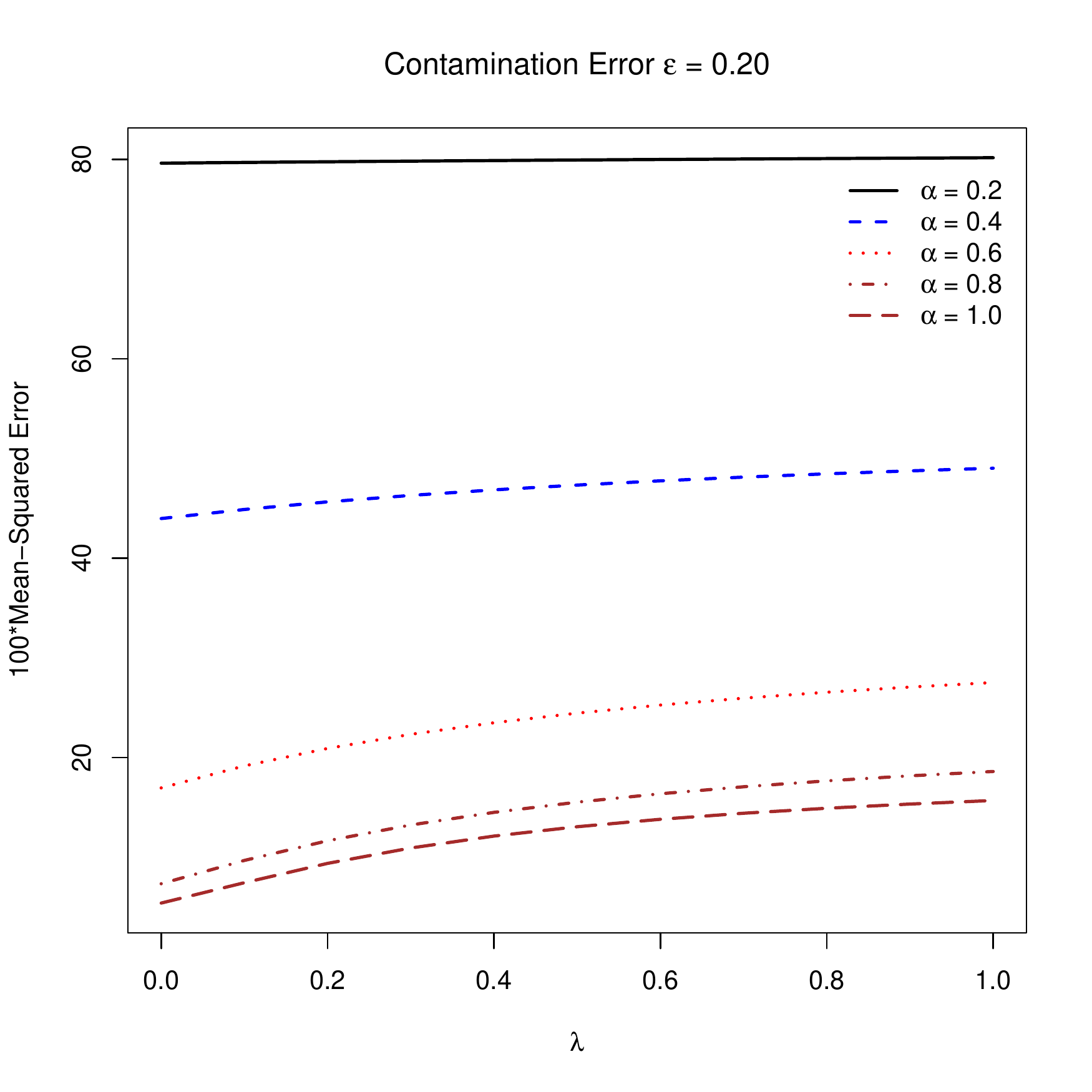}}
\end{minipage}
\caption{Scaled MSE of the Minimum Bridge Divergence Estimators for the $\mbox{Exp}(\sigma)$ model at $5\%$ and $20\%$ Contaminations.}
\label{mseepzeropttwol}
\end{figure}

For the data of size 20 presented in Section \ref{type0s} from $N(0,1)$, the minimum bridge divergence estimators of $\sigma$ under the $N(0, \sigma^2)$ model with $\alpha = 0.8$ and varying $\lambda$ from $0$ to $1$ are presented in Table \ref{tab:BDPDEst}. For each value of $\lambda$, the true global minimizers of the bridge divergence objective functions are presented in the parentheses. It is remarkable that even with $\lambda$ as high as 0.9, the spurious root phenomenon is observed. 
\begin{table}[!h]
\centering
\caption{Bridge Divergence Estimators of $\sigma$ with $\alpha = 0.8$. (The global minimizers are presented in 
the parentheses).}
\label{tab:BDPDEst}
\begin{tabular}{cccccc}\hline
&&&&&\\
$\lambda$ & $\hat{\theta}_{\alpha, \lambda}$ & $\lambda$ & $\hat{\theta}_{\alpha, \lambda}$ & $\lambda$ & $\hat{\theta}_{\alpha, \lambda}$ \\\hline
$0.9$     & 1.245808                         & $0.6$     & 1.248243                         & $0.3$     & 1.252942                         \\
    & (1.4411e-5)                        &     & (1.4124e-5)                         &     & (1.4085e-5)                         \\\hline
$0.8$     & 1.246477                         & $0.5$     & 1.249401                         & $0.2$     & 1.255756                       \\
     & (1.4218e-5)                       &    &(1.4106e-5)                         &      & (1.4078e-5)                         \\\hline
$0.7$     & 1.247269                         & $0.4$     & 1.250926                         & $0.1$     & 1.259926                       \\
   & (1.4155e-5)                         &      &(1.4094e-5)                        &     & (1.4073e-5)                         \\\hline
\end{tabular}
\end{table}

\section{Choice of Tuning Parameters}\label{tunpar}
Most of the robust estimators derive their robustness from down-weighting the observations suspected as outliers. As can be seen from the BDPD objective functions, the observations are proportionally weighted by powers of the model density which demonstrates that the observations are down-weighted for being outliers with respect to the density $f_{\theta}$. It is this outlier stability property which makes the corresponding estimators desirable from the point of view of robustness. If there are no outliers in the data, then the tuning parameter $\alpha$ must ideally be set to zero in order to get optimal inference. But if the data do contain outliers, then the tuning parameters $\alpha$ and $\lambda$ should be chosen so as to partially or fully eliminate the effect of the outlying observations. However the choice of the tuning parameter must be an automatic procedure where a data-driven choice of the parameters $\alpha$ and $\lambda$ is generated by the method. One of the first methods of choosing tuning parameters in case of the DPD was proposed in \cite{HK01}, where the tuning parameter is chosen by minimizing the trace of an estimate of the asymptotic covariance matrix. \cite{wj05} refined this method to find an ``optimal" tuning parameter by minimizing the trace of an estimate of the MSE of the estimator. Here the MSE is computed by separately estimating the bias and the variance components. From the asymptotic variance formula, one can easily estimate the variance. But bias estimation involves the use of a pilot estimate, and the estimated mean square error has a strong dependence on it. \cite{wj05} used the minimum $L_2$ distance estimator as the pilot; however, under the current state of the art, there is no universally acceptable choice of the pilot, on which the process critically depends. 

We acknowledge that the method of \cite{wj05} would be perfect if we could eliminate its dependence on the pilot estimator, or find an independent estimate of the bias as a function of the tuning parameter. However, here we present a modified version of the approach of \cite{HK01} and give an (almost theoretical) justification of the same. Firstly, one should not use the trace of an estimate of the asymptotic covariance matrix (although some of the present authors have done so in the past in the absence of a better strategy) for the construction of the objective function to be minimized. This leads to adding quantities with different units. Consider the $N(\mu,\sigma^2)$ example in which case the asymptotic variance of $(\hat{\mu}, \hat{\sigma}^2)$ or any of the variance estimates have the units of $X_i$ and $X_i^2$ that should not be added. One could, of course, think of estimating $(\mu,\sigma)$ where adding the variance estimates would be more sensible. But this remedy does not work in general parameter spaces. 

Why is the asymptotic variance adequate? An intuitive explanation is as follows. It is well-known that the delete-$d$ jack-knife estimator is a consistent estimator of the asymptotic variance. The delete-$d$ jack-knife estimator is obtained by calculating the difference between the estimator calculated based on $n$ observations and the estimator calculated based on $n-d$ observations. By the definition of a good robust estimator, this difference should be small for a robust estimator but for a non-robust estimator like the maximum likelihood estimator this difference would be large if the $d$ observations deleted contain some outliers. This implies that robust estimators should have ``smaller" asymptotic variance than non-robust estimators. Of course, if there is no contamination then it is known that the maximum likelihood estimator would have ``smaller" variance for large enough sample size. Using a bootstrap asymptotic variance estimator also gives a similar conclusion. We do ignore the bias, but the question clearly goes in favor of a robust estimator. By definition they are closer to the true parameter based on the majority of the data. 

A more concrete explanation is offered by using the closeness measure introduced in \cite{ZH16}. Let $\mathcal{V}_{\Lambda}$ define a set of variance matrices $V_{\alpha}$ indexed by a (possibly vector) parameter $\alpha\in\Lambda$.
\begin{defi}\label{def:Close}
A non-negative real-valued matrix function $\mathfrak{m}(\cdot)$ defined on the set $\mathcal{V}_{\Lambda}$ is called a \mbox{closeness measure} if and only if the following conditions hold.
\begin{enumerate}
\item Consistency. For any $V_1, V_2\in\mathcal{V}_{\Lambda}$, if $V_1 \succeq V_2$, then $\mathfrak{m}(V_1) \ge \mathfrak{m}(V_2)$, where the equality holds only if $V_1 = V_2$.
\item Continuity. For any matrix sequence $\{V_n\}\subseteq \mathcal{V}_{\Lambda}$ with $V_n \succeq B$. As $n\to\infty$, $\norm{V_n - B}_{\infty}\to0$ if and only if $\mathfrak{m}(V_n)\to\mathfrak{m}(B)$.
\end{enumerate}
\end{defi}
Here by $A\succeq B$, we mean $A - B$ is non-negative definite. Definition \ref{def:Close} implies that if there is an efficient variance matrix in the set $\mathcal{V}_{\Lambda}$, then minimizing $\mathfrak{m}(V_{\alpha})$ over $\alpha\in\Lambda$ leads to such a matrix. \cite{ZH16} additionally prove that the trace and the Frobenious norm are valid closeness measures. We use the determinant of the matrix as the measure of closeness (see Section S.5 of the supplementary material). Note that determinant is a valid quantity to consider from the point of view of units and it also has a practical interpretation as the volume of a confidence set. These arguments give a justification for minimizing a closeness measure of an estimate of the asymptotic variance. 

Summing up all these arguments, we claim that minimizing the determinant of an estimate of the asymptotic variance provides an asymptotically optimal estimator whenever such an estimator exists in the family of estimators under consideration. A detailed analysis of this procedure would be taken up in a future paper. Just as a remark, for the normal data of size 20 presented in Section \ref{type0s}, the optimal $\lambda$ parameter with $\alpha = 0.8$ using this procedure turned out to be $\lambda = 1$. A plot of the asymptotic variance over all $\lambda$ is given in the supplementary material.
\section{Concluding Remarks}\label{con}
In this paper, the competing families of divergences, DPD and LDPD, are critically examined for their strengths and deficiencies. The bridge divergence family introduced in this paper is an attempt at combining the good of both and nullify the disadvantages of either. Unlike the DPD, the LDPD estimating equation admits roots with small latent bias. However, these roots may not be the global minimizers of the LDPD objective. The phenomenon of spurious global minimizers of the LDPD is rigorously proved in specific parametric families where the DPD provably works.  The bridge divergences partly suppresses this problem for certain tuning parameter ($\lambda$) values along the bridge.  

However, the Bridge divergence is not a perfect solution in that it also faces the same problem as LDPD in some cases. The point made here is that one cannot expect the global minimizer of the LDPD to generate a good estimator and the DPD estimator is a safe bet in all the cases, but whenever possible some members of the minimum bridge divergence estimators can help reduce the latent bias of the DPD estimator. 
\section*{Acknowledgements}
The authors would like to thank Srijata Samanta of University of Florida for her contribution towards Remark \ref{rem:ss}.
\bibliographystyle{apalike}
\bibliography{bdiver}

\begin{thebibliography}{}

\bibitem[Basu et~al., 1998]{bhhj98}
Basu, A., Harris, I.~R., Hjort, N.~L., and Jones, M.~C. (1998).
\newblock Robust and efficient estimation by minimising a density power
  divergence.
\newblock {\em Biometrika}, 85(3):549--559.

\bibitem[Bickel and Doksum, 2015]{BD15}
Bickel, P.~J. and Doksum, K.~A. (2015).
\newblock {\em Mathematical statistics---basic ideas and selected topics.
  {V}ol. 1}.
\newblock Texts in Statistical Science Series. CRC Press, Boca Raton, FL,
  second edition.

\bibitem[Broniatowski et~al., 2012]{bt12}
Broniatowski, M., Toma, A., and Vajda, I. (2012).
\newblock Decomposable pseudodistances and applications in statistical
  estimation.
\newblock {\em J. Statist. Plann. Inference}, 142(9):2574--2585.

\bibitem[Ferguson, 1996]{Ferguson}
Ferguson, T. (1996).
\newblock {\em A Course in Large Sample Theory}.
\newblock Chapman \& Hall Texts in Statistical Science Series. Taylor \&
  Francis.

\bibitem[Fujisawa, 2013]{fuji13}
Fujisawa, H. (2013).
\newblock Normalized estimating equation for robust parameter estimation.
\newblock {\em Electron. J. Stat.}, 7:1587--1606.

\bibitem[Fujisawa and Eguchi, 2008]{fe08}
Fujisawa, H. and Eguchi, S. (2008).
\newblock Robust parameter estimation with a small bias against heavy
  contamination.
\newblock {\em J. Multivariate Anal.}, 99(9):2053--2081.

\bibitem[Hong and Kim, 2001]{HK01}
Hong, C. and Kim, Y. (2001).
\newblock Automatic selection of the tuning parameter in the minimum density
  power divergence estimation.
\newblock {\em J. Korean Statist. Soc.}, 30(3):453--465.

\bibitem[Jones et~al., 2001]{jhhb01}
Jones, M.~C., Hjort, N.~L., Harris, I.~R., and Basu, A. (2001).
\newblock A comparison of related density-based minimum divergence estimators.
\newblock {\em Biometrika}, 88(3):865--873.

\bibitem[Lehmann and Casella, 1998]{LEH}
Lehmann, E.~L. and Casella, G. (1998).
\newblock {\em Theory of point estimation}.
\newblock Springer Texts in Statistics. Springer-Verlag, New York, second
  edition.

\bibitem[Lindsay, 1994]{LIND94}
Lindsay, B.~G. (1994).
\newblock Efficiency versus robustness: the case for minimum {H}ellinger
  distance and related methods.
\newblock {\em Ann. Statist.}, 22(2):1081--1114.

\bibitem[van~der Vaart, 1998]{VAAR98}
van~der Vaart, A.~W. (1998).
\newblock {\em Asymptotic statistics}.
\newblock Cambridge University Press, Cambridge.

\bibitem[Warwick and Jones, 2005]{wj05}
Warwick, J. and Jones, M.~C. (2005).
\newblock Choosing a robustness tuning parameter.
\newblock {\em J. Stat. Comput. Simul.}, 75(7):581--588.

\bibitem[Windham, 1995]{wind95}
Windham, M.~P. (1995).
\newblock Robustifying model fitting.
\newblock {\em J. Roy. Statist. Soc. Ser. B}, 57(3):599--609.

\bibitem[Zhang and He, 2016]{ZH16}
Zhang, S. and He, X. (2016).
\newblock Inference based on adaptive grid selection of probability transforms.
\newblock {\em Statistics}, 50(3):667--688.

\end{thebibliography}


\begin{thebibliography}{}

\bibitem[Ferguson, 1996]{Ferguson}
Ferguson, T. (1996).
\newblock {\em A Course in Large Sample Theory}.
\newblock Chapman \& Hall Texts in Statistical Science Series. Taylor \&
  Francis.

\bibitem[van~der Vaart, 1998]{VAAR98}
van~der Vaart, A.~W. (1998).
\newblock {\em Asymptotic statistics}.
\newblock Cambridge University Press, Cambridge.

\end{thebibliography}
\end{document}


\maketitle
\section{Consistency of the Minimum Bridge Divergence Estimator}
In this section, we give a proof of Theorem 3.1 of the main paper. To begin with, let us define two quantities:
\begin{align*}
M_n(\theta) &= \frac{1}{\bar{\lambda}}\log \left(\lambda + \bar{\lambda} \int f_\theta^{1+\alpha}\right) - \frac{1}{\bar{\lambda}}\left(\frac{1+\alpha}{\alpha}\right)\log \left(\lambda + \bar{\lambda} \frac{1}{n}\sum_{i=1}^{n}f_\theta^\alpha(X_i)\right),\\
M(\theta) &= \frac{1}{\bar{\lambda}}\log \left(\lambda + \bar{\lambda} \int f_\theta^{1+\alpha}\right) - \frac{1}{\bar{\lambda}}\left(\frac{1+\alpha}{\alpha}\right)\log \left(\lambda + \bar{\lambda} \int g f_\theta^{\alpha}\right).
\end{align*}
We have the estimator and the target as 
\[
\hat{\theta}_n^{(\alpha,\lambda)} = \argmin_{\theta\in\Theta} M_n(\theta)\quad\mbox{and}\quad
\theta_g^{(\alpha,\lambda)} = \argmin_{\theta\in\Theta} M(\theta).
\]
In order to conclude that $\hat{\theta}_n^{(\alpha,\lambda)} \overset{a.s.}{\to} \theta_g^{(\alpha,\lambda)}$, we apply Theorem 5.7 of \cite{VAAR98}. To this end, we prove almost sure uniform convergence of the (random) functions $M_n(\theta)$ to $M(\theta)$.
Note that
\begin{align*}
\sup_{\theta \in \Theta} |M_n(\theta) - M(\theta)| &= \frac{1}{\bar{\lambda}}\left(\frac{1+\alpha}{\alpha}\right) \sup_{\theta \in \Theta} \left|\log \left(\lambda + \bar{\lambda} \frac{1}{n}\sum_{i=1}^{n}f_\theta^\alpha(X_i)\right) - \log \left(\lambda + \bar{\lambda} \int g f_\theta^{\alpha}\right)\right|\\
&\leq \frac{1}{\lambda}\left(\frac{1+\alpha}{\alpha}\right)  \sup_{\theta \in \Theta} \left|\frac{1}{n}\sum_{i=1}^{n}f_\theta^\alpha(X_i) - \int g f_\theta^{\alpha}\right|.  
\end{align*}
The last inequality follows from the fact that if $\lambda \le a < b < \infty$, then there exists $\xi \in (a,b)$ such that $\log b - \log a = \frac{1}{\xi} (b-a) \leq \frac{1}{\lambda} (b-a)$. By (C2), we know that $\int gf_{\theta}^{\alpha} \le \int g(x)K(x)dx < \infty$ for all $\theta$. Applying Theorem $16$(a) of \cite{Ferguson} with $U(x,\theta) = f_\theta^\alpha(x)$, we conclude that the quantity on the right hand side above converges almost surely to zero. Note that the conditions $(1), (2)$ and $(3)$ of Theorem $16$(a) of \cite{Ferguson} are exactly the same as our conditions (C1),(C3) and (C2), respectively, and we are done.    
\begin{remark}
For the above proof to hold, the assumption $\lambda > 0$ is crucial to bound the terms $\lambda + \bar{\lambda} \frac{1}{n}\sum_{i=1}^{n}f_\theta^\alpha(X_i)$ and $\lambda + \bar{\lambda} \int g f_\theta^{\alpha}$ away from $0$, so that the mean value theorem can be applied on the $\log$ function over its domain $(0,\infty)$. Hence the proof, as it is, cannot be applied on the LDPD, but it works for every other fixed member of the bridge family.  
\end{remark}
\section{Unboundedness of the LDPD Objective based on Simulated Data}\label{sec:SuppSp2}
In Section 5 of the main article, we plotted the LDPD objective with $\alpha = 0.5$ based on an artificial sample of size $4$, the underlying model being the family $\{N(\mu,\sigma^2):\mu \in \mathbb{R},\sigma>0\}$. Here, we demonstrate the unboundedness of the LDPD objective with $\alpha = 0.5$ based on 100 observations simulated from $N(0,1)$. The interactive MATLAB figure can be found in the file \verb|hundnormfinal.pdf| (a part of the online supplement), and the simulated data can be found in \verb|norhundfin.mat| (a part of the online supplement). 
\begin{figure}[H]\centering
	\resizebox{\textwidth}{!}{%
		\includegraphics{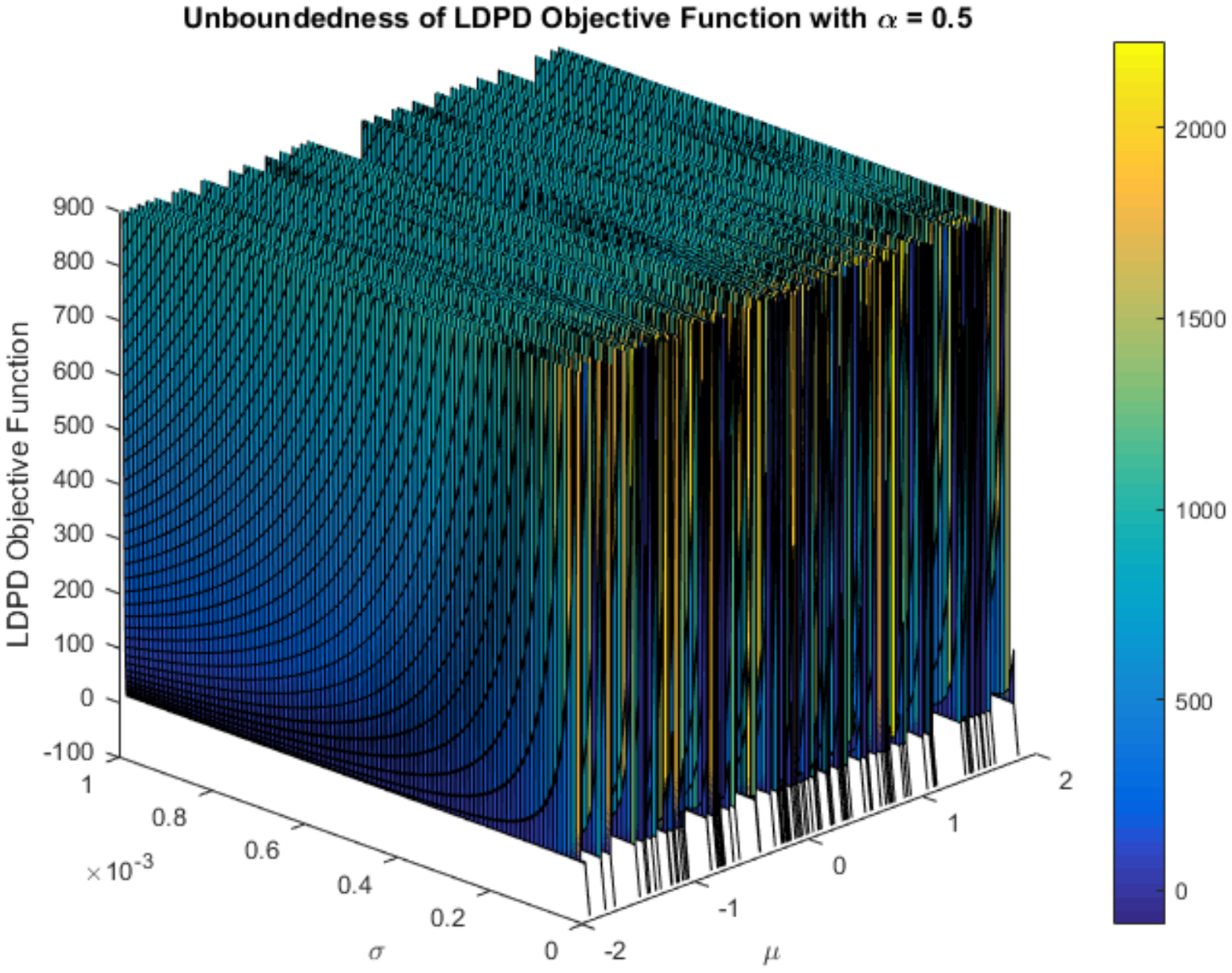}}\vspace{-50mm}
	\caption{LDPD Objective Based on 100 Observations Simulated From $N(0,1)$}
	\label{normalhundred}
\end{figure}  

The figure shows sharp dips of the LDPD objective when the mean parameter is one of the $100$ data points and the standard deviation approaches $0$. From the proof of Theorem 5.1 of the main article, we see that the Bridge divergence objective function goes to $+\infty$ as $\sigma$ goes to $0$ for all those $\mu$ which do not equal any data point, at a fast rate. That is why, the plot looks steep when $\mu$ equals a data point. However, putting $\lambda = 0$ in the same proof, we see that the LDPD objective function goes to $-\infty$ for all those $\mu$ which equal one of the data points, at rate $\log \sigma$, which is a slow rate. That is why, we need a very small starting value in the $\sigma$-axis, to observe the drop in the LDPD objective function.
\section{Bias and Mean-Squared Error of the Bridge Estimators}
In Secion 7 of the main paper, only the plots of the mean-squared errors of the bridge estimators for the exponential scale model case with $5\%$ and $20\%$ contaminations are included. We now present the exact numerical values of the (scaled up) bias and the (scaled up) mean-squared error of the bridge estimators (in Tables 1 to 6) for all the three levels of contamination: $0\% , 5\%$ and $20\%$, for both the exponential and the normal scale family cases. To gain more numerically significant digits, we scaled up the bias and MSE by 10 ($=\sqrt{n}$) and 100 ($=n$), respectively. In each table, along each row, the bias/MSE of the bridge estimators is noted, as the chain algorithm proceeds from $\lambda=1$ to $\lambda=0$ in steps of $0.1$. 

In the normal scale family case with contamination level $0.2$, it is seen that for each $\lambda$, the maximum likelihood estimator (the estimator with $\alpha = 0$) has the smallest mean squared error. The reason behind this is that in this case, we have an inner contamination around the mode of the majority distribution. 

To elaborate, consider a weighted estimating equation of the form:
\begin{equation}\label{besteq}
	\sum_{i=1}^n w_i^{\alpha} u_\theta(X_i) = 0.
\end{equation}
Suppose that $w_i=f_{\eta}(X_i)$, where $f_\theta$ is the density function of $N(5,\theta^2)$, and $\eta$ is an initial value of $\theta$. Suppose that the true value of $\theta$ is $1$. For $\alpha > 0$, $w_i^\alpha$ is large if $X_i$ is close to $5$. So, observations coming from the contaminating distribution get more weight relative to observations from the true majority distribution. Consequently, the solution $$\hat{\theta}^{(\alpha)}= \sqrt{\frac{\sum_{i=1}^n w_i^{\alpha} (X_i-5)^2}{\sum_{i=1}^n w_i^{\alpha}}}$$ to (S.3.1) shrinks more towards $0$ when $\alpha > 0$, than when $\alpha = 0$.Thus, the mean squared error around $1$, i.e. $\mathbb{E}(\hat{\theta}^{(\alpha)} - 1)^2$ is inversely related to shrinkage towards $0$. That is why, we should expect $$\mathbf{MSE}\left(\hat{\theta}^{(\alpha)}\right) \gg \mathbf{MSE}\left(\hat{\theta}^{(0)}\right)$$ for $\alpha > 0$. The estimating equation (S.3.1) is similar to that of the DPD. Similar heuristic argument also applies to the other bridge divergences, as they downweight using powers of densities.

\begin{sidewaystable}[!h]\centering
\captionsetup{width = 0.8\textwidth}
		\caption{Bias of the minimum bridge divergence estimators for $\epsilon = 0$ in the exponential scale family case. The mean squared errors are given in the parentheses.}\label{expbz}
	\begin{tabular}{|c|ccccccccccc|}\hline
		\backslashbox{$\alpha$}{$\lambda$}   & 1.0      & 0.9       & 0.8      & 0.7       & 0.6       & 0.5       & 0.4       & 0.3       & 0.2       & 0.1       & 0.0       \\\hline
		0.0 &-0.0286 & -0.0286 & -0.0286 & -0.0286 & -0.0286 & -0.0286 & -0.0286 & -0.0286 & -0.0286 & -0.0286 & -0.0286 \\
		  & (0.9963) & (0.9963) & (0.9963) & (0.9963) & (0.9963) & (0.9963) & (0.9963) & (0.9963) & (0.9963) & (0.9963) & (0.9963) \\
		0.2 &-0.0132 & -0.0132 & -0.0132 & -0.0132 & -0.0132 & -0.0132 & -0.0132 & -0.0132 & -0.0132 & -0.0131 & -0.0131 \\
	    	& (1.0915) & (1.0918) & (1.0922) & (1.0925) & (1.0929) & (1.0933) & (1.0937) & (1.0942) & (1.0946) & (1.0951) & (1.0956) \\
		0.4 &-0.0050 & -0.0051 & -0.0051 & -0.0052 & -0.0052 & -0.0052 & -0.0052 & -0.0053 & -0.0052 & -0.0052 & -0.0051 \\
		    & (1.2918) & (1.2950) & (1.2984) & (1.3021) & (1.3063) & (1.3108) & (1.3158) & (1.3214) & (1.3277) & (1.3348) & (1.3428) \\
		0.6 & 0.0028  & 0.0026  & 0.0023  & 0.0020  & 0.0017  & 0.0014  & 0.0010  & 0.0008  & 0.0006  & 0.0005  & 0.0007  \\
			& (1.5066) & (1.5155) & (1.5257) & (1.5372) & (1.5505) & (1.5660) & 1.5842 & 1.6060 & 1.6324 & 1.6651 & 1.7068 \\
		0.8 & 0.0113  & 0.0107  & 0.0101  & 0.0094  & 0.0086  & 0.0077  & 0.0068  & 0.0057  & 0.0047  & 0.0041  & 0.0047  \\
			& (1.6968) & (1.7123) & (1.7304) & (1.7517) & (1.7773) & (1.8085) & (1.8473) & (1.8970) & (1.9627) & (2.0535) & (2.1873) \\
		1.0 & 0.0198  & 0.0190  & 0.0180  & 0.0168  & 0.0155  & 0.0138  & 0.0119  & 0.0095  & 0.0068  & 0.0044  & 0.0056\\ 
			& (1.8563) & (1.8770) & (1.9016) & (1.9315) & (1.9686) & (2.0156) & (2.0771) & (2.1609) & (2.2813) & (2.4684) & (2.7981)\\\hline
	\end{tabular}
\end{sidewaystable}

\begin{sidewaystable}\centering
\captionsetup{width = 0.8\textwidth}
	\caption{Bias of the minimum bridge divergence estimators for $\epsilon = 0.05$ in the exponential scale family case. The mean squared errors are given in the parentheses.}\label{expbf}
	\begin{tabular}{|c|ccccccccccc|}\hline
		\backslashbox{$\alpha$}{$\lambda$}   & 1      & 0.9       & 0.8      & 0.7       & 0.6       & 0.5       & 0.4       & 0.3       & 0.2       & 0.1       & 0.0       \\\hline
0.0& 2.5006 & 2.5006 & 2.5006 & 2.5006 & 2.5006 & 2.5006 & 2.5006 & 2.5006 & 2.5006 & 2.5006 & 2.5006 \\
& (8.3365) & (8.3365) & (8.3365) & (8.3365) & (8.3365) & (8.3365) & (8.3365) & (8.3365) & (8.3365) & (8.3365) & (8.3365) \\
0.2& 1.5970 & 1.5954 & 1.5937 & 1.5920 & 1.5901 & 1.5882 & 1.5861 & 1.5840 & 1.5817 & 1.5793 & 1.5768 \\
& (4.5538) & (4.5490) & (4.5439) & (4.5385) & (4.5329) & (4.5270) & (4.5208) & (4.5142) & (4.5072) & (4.4999) & (4.4921) \\
0.4& 0.9389 & 0.9306 & 0.9215 & 0.9117 & 0.9008 & 0.8888 & 0.8756 & 0.8608 & 0.8442 & 0.8256 & 0.8044 \\
& (2.8677) & (2.8564) & (2.8441) & (2.8307) & (2.8162) & (2.8002) & (2.7826) & (2.7632) & (2.7417) & (2.7177) & (2.6907) \\
0.6& 0.6599 & 0.6459 & 0.6301 & 0.6122 & 0.5918 & 0.5682 & 0.5406 & 0.5080 & 0.4690 & 0.4215 & 0.3626 \\
& (2.4616) & (2.4543) & (2.4465) & (2.4382) & (2.4294) & (2.4202) & (2.4105) & (2.4007) & (2.3911) & (2.3826) & (2.3766) \\
0.8& 0.5923 & 0.5769 & 0.5590 & 0.5382 & 0.5134 & 0.4836 & 0.4471 & 0.4015 & 0.3428 & 0.2652 & 0.1585 \\
& (2.4580) & (2.4572) & (2.4569) & (2.4574) & (2.4595) & (2.4637) & (2.4715) & (2.4853) & (2.5093) & (2.5517) & (2.6287) \\
1.0& 0.6000 & 0.5857 & 0.5689 & 0.5487 & 0.5240 & 0.4933 & 0.4540 & 0.4021 & 0.3307 & 0.2272 & 0.0670\\
& (2.5751) & (2.5800) & (2.5865) & (2.5955) & (2.6080) & (2.6261) & (2.6532) & (2.6956) & (2.7659) & (2.8918) & (3.1419)\\\hline
	\end{tabular}
\end{sidewaystable}

\begin{sidewaystable}\centering
\captionsetup{width = 0.8\textwidth}
	\caption{Bias of the minimum bridge divergence estimators for $\epsilon = 0.2$ in the exponential scale family case. The mean squared errors are given in the parentheses.}\label{expbt}
	\begin{tabular}{|c|ccccccccccc|}\hline
		\backslashbox{$\alpha$}{$\lambda$}   & 1      & 0.9       & 0.8      & 0.7       & 0.6       & 0.5       & 0.4       & 0.3       & 0.2       & 0.1       & 0.0       \\\hline
		0.0& 9.9310 & 9.9310 & 9.9310 & 9.9310 & 9.9310 & 9.9310 & 9.9310 & 9.9310 & 9.9310 & 9.9310 & 9.9310 \\
		& (103.1433) & (103.1433) & (103.1433) & (103.1433) & (103.1433) & (103.1433) & (103.1433) & (103.1433) & (103.1433) & (103.1433) & (103.1433) \\
		0.2& 8.5802 & 8.5776 & 8.5748 & 8.5719 & 8.5688 & 8.5655 & 8.5619 & 8.5581 & 8.5541 & 8.5496 & 8.5449 \\
		& (80.1781)  & (80.1383)  & (80.0956)  & (80.0508)  & (80.0029)  & (79.9513)  & (79.8964)  & (79.8376)  & (79.7743)  & (79.7059)  & (79.6319)  \\
		0.4& 6.3748 & 6.3504 & 6.3232 & 6.2927 & 6.2583 & 6.2192 & 6.1744 & 6.1225 & 6.0618 & 5.9898 & 5.9034 \\
		& (49.0325)  & (48.7654)  & (48.4686)  & (48.1369) & (47.7637) & (47.3413)  & (46.8588)  & (46.3030)  & (45.6555)  & (44.8928)  & (43.9815)  \\
		0.6& 4.4240 & 4.3646 & 4.2958 & 4.2153 & 4.1202 & 4.0060 & 3.8666 & 3.6934 & 3.4734 & 3.1869 & 2.8039 \\
		& (27.5293)  & (27.0756)  & (26.5580)  & (25.9626)  & (25.2711)  & (24.4596)  & (23.4958)  & (22.3361)  & (20.9208)  & (19.1671)  & (16.9600)  \\
		0.8& 3.5066 & 3.4379 & 3.3566 & 3.2591 & 3.1402 & 2.9922 & 2.8036 & 2.5565 & 2.2216 & 1.7511 & 1.0688 \\
		& (18.6063)  & (18.1708)  & (17.6661)  & (17.0746)  & (16.3735)  & (15.5316)  & (14.5061)  & (13.2397)  & (11.6579)  & (9.6848)   & (7.3460)   \\
		1.0& 3.2200 & 3.1591 & 3.0861 & 2.9967 & 2.8852 & 2.7422 & 2.5529 & 2.2921 & 1.9147 & 1.3370 & 0.4157\\
		& (15.6965)  & (15.3395)  & (14.9194)  & (14.4182)  & (13.8115)  & (13.0645)  & (12.1278)  & (10.9318)  & (9.3892)   & (7.4499)   & (5.4069) \\\hline
	\end{tabular}
\end{sidewaystable}

\begin{sidewaystable}\centering
\captionsetup{width = 0.8\textwidth}
	\caption{Bias of the minimum bridge divergence estimators for and $\epsilon = 0$ in the normal scale family case. The mean squared errors are given in the parentheses.}\label{normbz}
	\begin{tabular}{|c|ccccccccccc|}\hline
		\backslashbox{$\alpha$}{$\lambda$}   & 1      & 0.9       & 0.8      & 0.7       & 0.6       & 0.5       & 0.4       & 0.3       & 0.2       & 0.1       & 0.0       \\\hline
		0.0& -0.0191 & -0.0191 & -0.0191 & -0.0191 &-0.0191 & -0.0191 & -0.0191 & -0.0191 & -0.0191 & -0.0191 & -0.0191 \\
		& (2.1467) & (2.1467) & (2.1467) & (2.1467) & (2.1467) & (2.1467) & (2.1467) & (2.1467) & (2.1467) & (2.1467) & (2.1467) \\
		0.2& -0.0263 & -0.0263 & -0.0263 & -0.0263 & -0.0263 & -0.0263 & -0.0263 & -0.0263 & -0.0262 & -0.0262 & -0.0262 \\
		& (0.5846) & (0.5846) & (0.5847) & (0.5848) & (0.5849) & (0.5850) & (0.5851) & (0.5852) & (0.5853) & (0.5854) & (0.5855) \\
		0.4& -0.0266 & -0.0267 & -0.0267 & -0.0268 & -0.0269 & -0.0269 & -0.0270 & -0.0271 & -0.0272 & -0.0273 & -0.0274 \\
		& (0.6812) & (0.6818)) & (0.6826) & (0.6834) & (0.6843) & (0.6853) & (0.6864) & (0.6877) & (0.6893) & (0.6910) & (0.6931) \\
		0.6& -0.0294 & -0.0296 & -0.0298 & -0.0301 & -0.0304 & -0.0307 & -0.0311 & -0.0315 & -0.0321 & -0.0327 & -0.0335 \\
		& (0.7944) & (0.7964) & (0.7987) & (0.8013) & (0.8043) & (0.8080) & (0.8124) & (0.8178) & (0.8245) & (0.8332) & (0.8449) \\
		0.8& -0.0309 & -0.0313 & -0.0317 & -0.0323 & -0.0329 & -0.0337 & -0.0346 & -0.0359 & -0.0374 & -0.0395 & -0.0423 \\
		& (0.9041) & (0.9078) & (0.9121) & (0.9172) & (0.9234) & (0.9312) & (0.9410) & (0.9540) & (0.9718) & (0.9978) & (1.0393) \\
		1.0& -0.0302 & -0.0307 & -0.0313 & -0.0321 & -0.0331 & -0.0343 & -0.0358 & -0.0379 & -0.0410 & -0.0455 & -0.0525\\
		& (0.9987) & (1.0037) & (1.0097) & (1.0170) & (1.0262) & (1.0380) & (1.0538) & (1.0760) & (1.1093) & (1.1647) & (1.2741)\\\hline
	\end{tabular}
\end{sidewaystable}

\begin{sidewaystable}\centering
\captionsetup{width = 0.8\textwidth}
	\caption{Bias of the minimum bridge divergence estimators for $\epsilon = 0.05$ in the normal scale family case. The mean squared errors are given in the parentheses.}\label{normbf}
	\begin{tabular}{|c|ccccccccccc|}\hline
		\backslashbox{$\alpha$}{$\lambda$}   & 1      & 0.9       & 0.8      & 0.7       & 0.6       & 0.5       & 0.4       & 0.3       & 0.2       & 0.1       & 0.0       \\\hline
		0.0& -0.5136 & -0.5136 & -0.5136& -0.5136& -0.5136& -0.5136 & -0.5136 & -0.5136 & -0.5136 & -0.5136 & -0.5136 \\
		& (2.3361) & (2.3361) & (2.3361) & (2.3361) & (2.3361) & (2.3361) & (2.3361) & (2.3361) & (2.3361) & (2.3361) & (2.3361) \\
		0.2& -0.3604 & -0.3605 & -0.3606 & -0.3606 & -0.3607 & -0.3608 & -0.3609 & -0.3610 & -0.3611 & -0.3612 & -0.3613 \\
		& (0.7260) & (0.7260) & (0.7262) & (0.7264) & (0.7265) & (0.7267) & (0.7269) & (0.7270) & (0.7272) & (0.7275) & (0.7277) \\
		0.4& -0.4476 & -0.4481 & -0.4486 & -0.4493 & -0.4500 & -0.4508 & -0.4516 & -0.4526 & -0.4537 & -0.4550 & -0.4564 \\
		& (0.9104) & (0.9118) & (0.9132) & (0.9149) & (0.9167) & (0.9187) & (0.9209) & (0.9235) & (0.9264) & (0.9297) & (0.9336) \\
		0.6& -0.5335 & -0.5350 & -0.5366 & -0.5385 & -0.5406 & -0.5432 & -0.5462 & -0.5498 & -0.5542 & -0.5597 & -0.5667 \\
		& (1.1352) & (1.1396) & (1.1448) & (1.1506) & (1.1574) & (1.1654) & (1.1750) & (1.1865) & (1.2008) & (1.2190) & (1.2427) \\
		0.8& -0.6092 & -0.6117 & -0.6147 & -0.6182 & -0.6223 & -0.6274 & -0.6338 & -0.6421 & -0.6530 & -0.6683 & -0.6911 \\
		& (1.3591) & (1.3679) & (1.3783) & (1.3907) & (1.4057) & (1.4242) & (1.4476) & (1.4782) & (1.5198) & (1.5794) & (1.6720) \\
		1.0& -0.6716 & -0.6750 & -0.6790 & -0.6839 & -0.6900 & -0.6977 & -0.7079 & -0.7219 & -0.7422 & -0.7742 & -0.8314\\
		& (1.5545) & (1.5673) & (1.5827) & (1.6016) & (1.6253) & (1.6558) & (1.6967) & (1.7543) & (1.8413) & (1.9851) & (2.2646)\\\hline
	\end{tabular}
\end{sidewaystable}

\begin{sidewaystable}\centering
\captionsetup{width = 0.8\textwidth}
	\caption{Bias of the minimum bridge divergence estimators for $\epsilon = 0.2$ in the normal scale family case. The mean squared errors are given in the parentheses.}\label{normbt}
	\begin{tabular}{|c|ccccccccccc|}\hline
		\backslashbox{$\alpha$}{$\lambda$}   & 1      & 0.9       & 0.8      & 0.7       & 0.6       & 0.5       & 0.4       & 0.3       & 0.2       & 0.1       & 0.0       \\\hline
		0.0& -1.9778 & -1.9778 & -1.9778 & -1.9778 &-1.9778 & -1.9778 & -1.9778 & -1.9778 & -1.9778 & -1.9778 & -1.9778 \\
		& (5.7870)  & (5.7870)  & (5.7870)  & (5.7870)  & (5.7870)  & (5.7870)  & (5.7870)  & (5.7870)  & (5.7870)  & (5.7870)  & (5.7870)  \\
		0.2& -2.7379 & -2.7381 & -2.7385 & -2.7387 & -2.7391 & -2.7394 & -2.7397 & -2.7401 & -2.7404 & -2.7408 & -2.7413 \\
		& (18.1322) & (18.1329) & (18.1339) & (18.1348) & (18.1358) & (18.1368) & (18.1378) & (18.1389) & (18.1401) & (18.1412) & (18.1426) \\
		0.4& -2.3770 & -2.3794 & -2.3821 & -2.3850 & -2.3882 & -2.3918 & -2.3957 & -2.4001 & -2.4050 & -2.4106 & -2.4170 \\
		& (11.0261) & (11.0370) & (11.0486) & (11.0614) & (11.0755) & (11.0910) & (11.1082) & (11.1275) & (11.1493) & (11.1739) & (11.2020) \\
		0.6& -2.2992 & -2.3084 & -2.3187 & -2.3307 & -2.3445 & -2.3608 & -2.3831 & -2.4134 & -2.4445 & -2.4831 & -2.5295 \\
		& (6.8577)  & (6.9119)  & (6.9739)  & (7.0458)  & (7.1303)  & (7.2318) & (7.4016)  & (7.6577)  & (7.8807)  & (8.1785)  & (8.5203)  \\
		0.8& -2.6605 & -2.6814 & -2.7062 & -2.7382 & -2.7941 & -2.8534 & -2.9425 & -3.0508 & -3.2000 & -3.4096 & -3.6874 \\
		& (8.6585)  & (8.8162)  & (9.0067)  & (9.2747)  & (9.8667)  & (10.4664) & (11.4595) & (12.6542) & (14.3434) & (16.7522) & (19.9901) \\
		1.0& -2.9865 & -3.0278 & -3.0800 & -3.1562 & -3.2333 & -3.3812 & -3.5566 & -3.8028 & -4.1377 & -4.6124 & -5.3932\\
		& (10.7030) & (11.1038) & (11.6395) & (12.5220) & (13.3224) & (15.1724) & (17.2957) & (20.3636) & (24.4783) & (30.3887) & (40.3177)\\\hline
	\end{tabular}
\end{sidewaystable} 
\section{Asymptotic Variance Plot of the Bridge Estimators}
The plot of the estimated asymptotic variance of the bridge estimators across $\lambda$ for $\alpha = 0.8$, using the normal data of size $20$ presented in Section 4 of the main article, is given in Figure 2. The quantity plotted is the estimated asymptotic variance using the formula given in Theorem 3.2, evaluated at $\theta=\hat{\theta}^{(\alpha,\lambda)}$ obtained through the chain algorithm for these data. It is seen that the estimated asymptotic variance decreases monotonically with $\lambda$. 

We also plot the estimated asymptotic variance of the bridge estimators across both $\lambda$ and $\alpha$, based on the normal data of size $20$ presented in Section 4 of the main article  (in Figure 3), and based on a sample of size $1000$ simulated from $N(0,1)$ (in Figure 4). The interactive MATLAB figures can be found in the files \verb|twodplotfinal.fig| and \verb|thousand.fig|, and the simulated sample of size $1000$ can be found in \verb|x.mat| (which are parts of the online supplement). The plot in Figure 2 represents the cross section of the plot in Figure 3 along $\alpha = 0.8$. 

\begin{figure}[H]\centering  
\resizebox{\textwidth}{!}{%
	\includegraphics{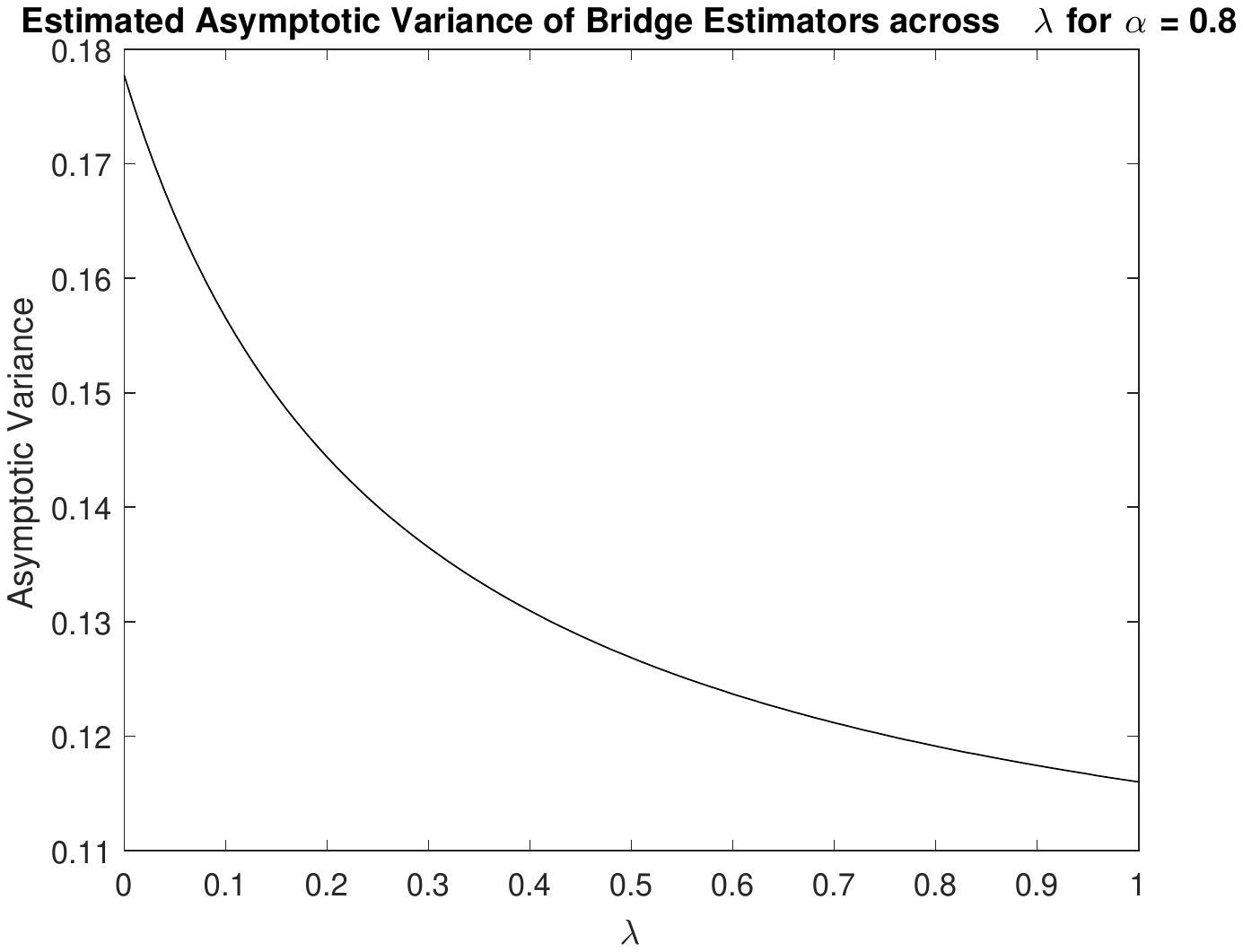}}\vspace{-62mm}
\caption{Plot of Estimated Asymptotic Variance of Bridge Estimators across $\lambda$ for $\alpha=0.8$, based on the normal data of size $20$}
\label{fig:tunone}
\end{figure}


\begin{figure}[H]\centering  
\resizebox{\textwidth}{!}{%
	\includegraphics{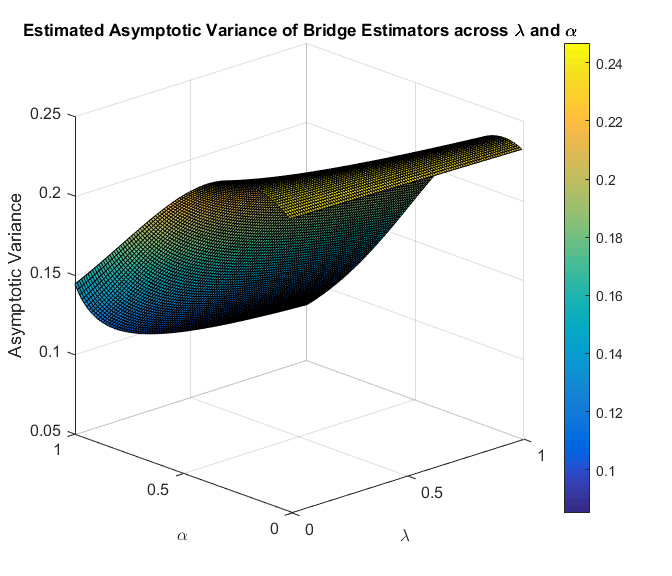}}
\caption{Plot of Estimated Asymptotic Variance of Bridge Estimators across $\lambda$ and $\alpha$, based on the normal data of size $20$}
\label{fig:twenty}
\end{figure}
\begin{figure}[H]\centering
\resizebox{\textwidth}{!}{%
	\includegraphics{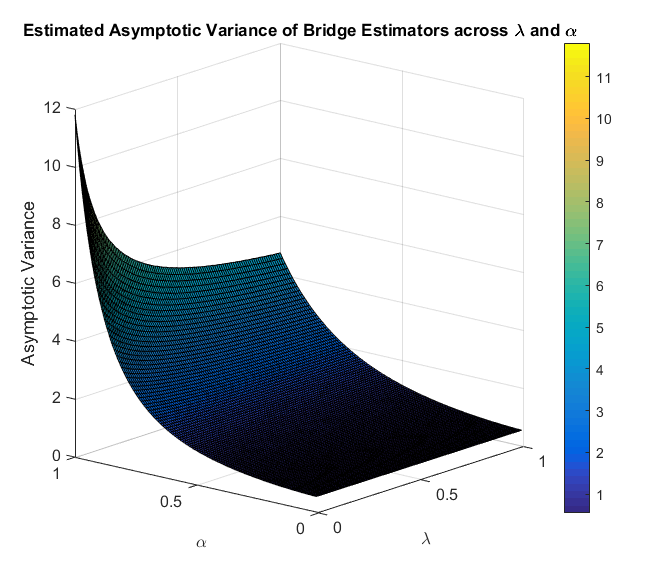}}
\caption{Plot of Estimated Asymptotic Variance of Bridge Estimators across $\lambda$ and $\alpha$, based on the simulated normal data of size $1000$}
\label{fig:thous}
\end{figure}

\section{The Determinant is a closeness measure}
\begin{thm}\label{thm:close}
The determinant of a matrix is a closeness measure.
\end{thm}
\begin{proof} Before getting into the proof, note that for any two positive definite matrices $A$ and $B$ of order $p\times p$ with a positive semi-definite difference $A - B$, we have
\[
|A| = |A - B + B| = |B|\left|I + B^{-1/2}(A - B)B^{-1/2}\right| = |B|\prod_{i=1}^p\left(1 + \lambda_i\right),
\]
where $\lambda_i \ge 0$ for all $1\le i\le p$ represents the eigenvalues of $B^{-1/2}(A - B)B^{-1/2}.$
 
(Consistency). The equality above implies that $|A| \ge |B|$ for $A\ge B$. To show that equality $|A| = |B|$ for $A\ge B$ holds if and only if $A = B$, first note that $A = B$ will trivially imply that $|A| = |B|$. If $|A| = |B|$, then by the equality above, we get $\prod_{i=1}^p(1 + \lambda_i) = 1$ which implies that $\lambda_i = 0$ for all $1\le i\le p$.

(Continuity). To show that $\|A_n - B\|_{\infty} \to 0$ if and only if $|A_n| \to |B|$ as $n\to\infty$ under the assumption that $A_n \ge B$ for all $n$, first note that $\|A_n - B\|_{\infty}$ converging to zero as $n\to\infty$ implies that $|A_n| - |B|\to0.$ This is because, the determinant is a continuous function of the elements of the matrix. So, it is now enough to prove the opposite direction. For proving this, define $\lambda_{in}, 1\le i\le p$ as the eigenvalues of $B^{-1/2}(A_n - B)B^{-1/2}$. Then by the equality above, we get that as $n\to\infty$,
\[
\prod_{i=1}^p (1 + \lambda_{in}) \to 1.
\]
This ensures that $\lambda_{in}\to 0$ as $n\to\infty$ for all $1\le i\le p$. Therefore $\|A_n - B\|_{\infty}$ converges to zero as $n\to\infty$.
\end{proof}

\bibliographystyle{apalike}
\bibliography{bdiver}